\documentclass[a4paper,twocolumn]{article}

\usepackage{authblk}
\usepackage{hyperref,breakurl}
\usepackage{amsmath,amsthm,amssymb}
\usepackage{epsfig}
\usepackage[left=1.40cm,top=1.4cm,right=1.40cm,bottom=1.4cm]{geometry}

\newtheorem{lemma}{Lemma}
\newtheorem{theorem}{Theorem}

\newtheorem{definition}{Definition}

\begin{document}
\title{On the Pricing of Recommendations and Recommending Strategically\footnote{This work was conducted as part of a EURYI scheme award. See {\protect \url{http://www.esf.org/euryi/}} for details.}}
%
% \numberofauthors{3}
%
\author[,1]{Paul D\"utting\footnote{Email: {\protect \url{paul.duetting@epfl.ch}}}}
\affil[1]{Ecole Polytechnique F\'ed\'erale de Lausanne, Lausanne, Switzerland}
\author[,2,3]{Monika Henzinger\footnote{Email: {\protect\url{monika.henzinger@univie.ac.at}}}}
\affil[2]{Google Switzerland, Z\"urich, Switzerland}
\affil[3]{University of Vienna, Vienna, Austria}
\author[,4]{Ingmar Weber \footnote{Email: {\protect \url{ingmar@yahoo-inc.com}}}}
\affil[4]{Yahoo! Research, Barcelona, Spain}

%\author{
%\alignauthor Paul D\"utting\\
%       \affaddr{EPFL Lausanne}\\
%       \affaddr{Lausanne, Switzerland}\\
%       \email{paul.duetting@epfl.ch}
%\alignauthor Monika Henzinger\\
%       \affaddr{University Vienna}\\
%       \affaddr{Vienna, Austria}\\
%       \email{henzinger@univie.ac.at}
%\alignauthor Ingmar Weber\\
%       \affaddr{Yahoo! Research}\\
%       \affaddr{Barcelona, Spain}\\
%       \email{ingmar@yahoo-inc.com}
%}
\date{\today}

\maketitle
\begin{abstract}
If you recommend a product to me and I buy it, how much should you be paid by the seller? And if your sole interest is to maximize the amount paid to you by the seller for a sequence of recommendations, how should you recommend optimally if I become more inclined to ignore you with each irrelevant recommendation you make? Finding an answer to these questions is a key challenge in all forms of marketing that rely on and explore social ties; ranging from personal recommendations to viral marketing.

In the first part of this paper, we show that there can be no pricing mechanism that is ``truthful'' with respect to the seller, and we use solution concepts from coalitional game theory, namely the Core, the Shapley Value, and the Nash Bargaining Solution, to derive provably ``fair'' prices for settings with one or multiple recommenders. We then investigate pricing mechanisms for the setting where recommenders have different ``purchase arguments''. Here we show that it might be beneficial for the recommenders to withhold some of their arguments, unless anonymity-proof solution concepts, such as the anonymity-proof Shapley value, are used.

In the second part of this paper, we analyze the setting where the recommendee loses trust in the recommender for each irrelevant recommendation. Here we prove that even if the recommendee regains her initial trust on each successful recommendation, the expected total profit the recommender can make over an infinite period is bounded. This can only be overcome when the recommendee also incrementally regains trust during periods without any recommendation. Here, we see an interesting connection to ``banner blindness'', suggesting that showing fewer ads can lead to a higher long-term profit.

% CHANGED:
% In the first part of this paper, we show that there can be no ``truthful'' pricing scheme to determine how much the recommender should be paid by the seller, and we use solution concepts from coalitional game theory, namely the Core and the Shapley Value, to derive provably ``fair'' pricing schemes for settings with one or multiple recommenders. We then investigate pricing schemes for the setting where recommenders have different ``purchase arguments''. Here we show that it might be beneficial for the recommenders to withhold some of their arguments, unless anonymity-proof solution concepts, such as the anonymity-proof Shapley value, are used.

% OLD VERSION:
%
% We study the problem of pricing recommendations: If $A$ recommends a 
% product from a seller $S$ to $B$ and $B$ buys it from $S$, how much 
% should $A$ be paid from $S$? Finding an answer to this question is a 
% key challenge in almost all forms of social media marketing ranging 
% from voluntary personal recommendations to viral marketing.
% 
% In this paper we show that there can be no ``truthful'' pricing scheme
% and derive``individual rational'' and ``fair'' pricing schemes based
% on coalitional game theory. We also investigate the realm of strategic
% behavior by the recommender and find that ...
%
% We see this as a modest, preliminary, exploratory work, whose main function
% is to point out, by dint of simplified and attractive examples, that the problem of 
% pricing recommendations can be usefully studied within a theoretical framework.

\vspace{6pt}\noindent{\bf Keywords:} recommendations, pricing mechanisms, trust loss in advertising, banner blindness
\end{abstract}
\section{Introduction}
Suppose you buy a new mobile phone and, simply because you like it so much, you recommend it to friends, encouraging them to purchase it as well. Even if you do not recommend it out of monetary reasons, what would be an adequate and \emph{fair} price for the phone manufacturer to pay for your recommendation?
If, on the other hand, you recommend a book at Amazon solely due to the monetary incentive given by Amazon's referral scheme\footnote{\url{https://affiliate-program.amazon.com/}} and your friends realize this, then they are likely to lose trust in your recommendations. Assuming your friends regain trust whenever you make a relevant recommendation, how can you maximize your long-term profit, and is this profit bounded or not? These are the two main research questions we address in this work.

%What is a recommendation worth? Mention social marketing platforms such as {\em Friend Vouch}~\cite{FV08}.
%``How much is your personal recommendation worth?'' ... brand ambassador

The importance of ``word-of-mouth'' (WOM) as a marketing channel has long been known \cite{BR87,HKK91,GCD03}. According to \cite{BR87}, ``WOM was seven times as effective as newspapers and magazines, four times as effective as personal selling, and twice as effective as radio advertising in influencing consumers to switch brands''. WOM is the causal effect behind ``brand congruence'' where friends both in offline \cite{RFBS84} and online \cite{SW09} social networks tend to use the same products. Recently, a platform called {\em Friend Vouch}~\cite{FV08} was founded around the idea of personal recommendations. Users of the service can become ``brand ambassadors'' who get paid for putting companies in touch with friends. Whether any personal touch is retained in such a system or whether the person in the middle is not simply another marketing channel is debatable and in Section \ref{sec:classification} we propose a classification schema to shed light on the differences.

As far as the pricing of recommendations is concerned, one could argue that honest recommendations should always be given without any monetary recompensation and that creating financial incentives could lead to a sell-out of friends. Although this is a valid concern, we argue that it might still be worth paying recommenders, even if these are not asking to be paid.
First, even though you might not be profit-maximizing in a strict sense you are probably more inclined to mention a certain product if there is some kind of recompensation: you might be honest enough not to recommend a bad product over a good one for financial reasons, but you are still more likely to recommend a good product if you get reimbursed.
Second, a \emph{fair} compensation can lead to increased brand loyalty. If you are already satisfied with the product then the feeling that the company recompensates you in a fair and adequate manner is likely to increase your positive attitude towards the company. On the other hand, if you are only offered \$1 for recommending a particular type of sports car then this might be viewed as ``offensive'' and is arguably worse than not being offered any recompensation.

Especially, as the issue of trust is of utmost relevance in the realm of personal recommendations, we believe that ``fair pricing'' is a cornerstone \cite{R04,M08a}.\footnote{Somewhat related is the phenomenon of pay-what-you-like pricing where people act ``irrationally'' and choose to pay an adequate amount \cite{M08b,FN09}.} Google e.g. advertises its Adsense program by claiming to use a Second-Price Auction to eliminate ``that feeling that you've paid too much''\footnote{\url{http://www.google.com/adsense/afs.pdf}}. In a similar spirit, our work on pricing recommendations can be viewed as trying to eliminate ``the feeling that you've been paid too little'' for your recommendation.

The second problem we study, relates to scenarios where the recommender is selfish and only makes paid recommendations to maximize her own profit. Here, in a sense, the friend making the recommendation is no more trustworthy or altruistic than a web search engine showing sponsored search results. In these settings we believe the trust between the recommender and the recommendee to be dissipating. More concretely we assume that with every unsuccessful recommendation the recommendee becomes more and more likely to ignore any ``advice'' given by her friend. We see this as closely related to ``banner blindness'' \cite{BL98,CHN00,BHNG05}, where people have become so overloaded and fed up with advertisement that they stop to notice it completely. Seen from this angle, our findings indicate that advertisers might have to \emph{stop} showing advertisements on a regular basis if they want to retain customers' trust without seeing click-through-rates converge to zero.

%\url{http://en.wikipedia.org/wiki/Referral_marketing}

% Different from referrals Amazon\footnote{\url{https://affiliate-program.amazon.com/gp/associates/join/landing/main.html}}

% Paid testimonial less effective than unpaid ones \cite{MM94}.

% Study of reviews left on Amazon \cite{CM06}.

%Maybe use the term relationship marketing \url{http://en.wikipedia.org/wiki/Relationship_marketing}. Reference for relationship marketing \cite{SP95}.

% Another term is ``testimonial'' \url{http://en.wikipedia.org/wiki/Testimonial}, essentially a not-personal recommendation.

%
\subsection{Related Work}

Even though recommendations can be seen as just another form of advertising, classical methods for the pricing of advertising, such as sponsored search auctions \cite{LPSV07}, are not directly applicable. This is mainly due to the fact that a true recommendation should be altruistic and so (i) the recommender is not profit maximizing and (ii) there is only a \emph{single} seller, as an altruistic recommender will not accept ``bids'' from multiple sellers. The differences between various kinds of advertising are described in %detail in 
Section \ref{sec:classification}.

The work that is most closely related to our paper is \cite{AMSX09}. There the authors study the sales price of an object as part of a viral marketing campaign. They assume that all ``converted'' nodes will try to convert all of their neighbors and that the conversion probability depends both on the number of neighbors converted and on the sales price. They do \emph{not} consider the problem of how the recommendation itself should be rewarded. In fact, they mention the problem of finding optimal ``cashbacks'' in settings where the nodes behave strategically as an open problem.

The problem of optimal pricing with non-social recommender systems, where the recommendations directly come from the potential seller, was studied in \cite{BO06}. Here by ``non-social'' we mean ``computer-generated'' and a typical example would be Amazon's ``Customers who bought X also bought Y''\footnote{\url{http://www.amazon.com}}. The somewhat surprising argument is that customers are willing to \emph{pay} for relevant recommendations as they create ``value by reducing product uncertainty for the customers''. In this paper, we consider the case where the recommendations are social and do not come from the seller directly. Though it is imaginable that the recommendee pays the recommender for a good recommendation, we do not investigate the pricing of this possible payment.

It should be clear that we are \emph{not} addressing the problem of \emph{what} to recommend, a problem typically encountered by stores such as Amazon and usually solved using ``collaborative filtering'' techniques \cite{SM95,HKTR04}. In the first part of this paper (Section~\ref{sec:pricing}), we assume that the recommender recommends an item because she believes this item to be of interest to the recommendee, and the algorithm used by her to determine potential interest is irrelevant. In the second part (Section~\ref{sec:recommending-strategically}), the recommender is profit maximizing and now only cares about the reward offered to her by the seller and the probability $p$ that the recommendee will buy the item. In this model the ``what'' is absorbed into $p$ and the recommender simply decides on \emph{when} to recommend. 

We are also \emph{not} addressing the topic of how rumors spread through social networks, or how to identify the best nodes to target for a viral marketing campaign \cite{DR01,KKT03}. Our work focuses on a single atomic link in the corresponding cascades of conversions and, in the first part, we ask what a fair price should be to pay a node for activating one of her neighbors. In answering this question we limit our attention to the immediate profit of the seller due to the individual sale, and we do \emph{not} consider the additional value due to recommendation cascades caused by the newly activated node. However, given any algorithm to compute this ``higher order'' profit, it can trivially be incorporated into our results.
The question whether a selfish node should actually try to activate her neighbors at all is addressed in Section~\ref{sec:recommending-strategically}.

%Related to the idea of cascades in viral marketing, there is the concept of ``social capital'' \cite{C94}. Here the question is, how much is \emph{your} social network worth? Here the value can be indirect, as e.g. in the case of your parents who decided to send you to a good school (so your parents are part of your social capital). Or it can be more direct as in the case of knowing both producers and customers of a certain product (which might let you profit from arbitrage). For a calculator that tries to put a concrete number on the value of your network, see \cite{MNV09}.

More generally, in the second part we look at a model where the recommendee loses trust in the recommender, i.e. for each unsuccessful recommendation she becomes less and less inclined to listen to any further suggestion. This is most likely to appear when the recommendee has the feeling that the recommendations are ``dishonest''. How honest recommendations can be ensured when there are several recommenders is studied in \cite{GEB04}. The approach suggested by the authors involves evaluating/ranking recommenders based on the rating given to their recommended items by other people. This motivates recommenders to give good recommendations in a similar way that Ebay's rating system gives incentives for both buyers and sellers ``to behave''. This approach, however, requires a public market where potential buyers can look for recommendations. This is not the setting of personal recommendations considered here.

The problem of trust decay is related to ``banner blindness'' \cite{BL98,CHN00,BHNG05}, where web users become ``blind'' to banner ads due to overexposure. Cast to this setting our mathematical model suggests that, even if web users' interest is ``refreshed'' by a single relevant advertisement that is clicked, the long term profit of advertisers will stagnate as click-through-rates fall to zero. The only possible way out of this dilemma is to \emph{stop} showing banner ads for a while so that users can ``unlearn'' to ignore all advertising. This approach is also suggested in a recent patent \cite{PK08}.

In typical literature on sponsored search auctions \cite{LPSV07,L06} it is assumed that the web search engine is optimizing its expected profit and that its expected profit for showing a particular ad is the ad's click-through-rate (CTR) multiplied by the price the advertiser will be charged when her ad gets clicked. Usually, only a single round is considered or, when there are budget constraints \cite{AMT07,FMPS07}, the CTRs are assumed to be \emph{constant} during the duration of the game. If, however, it is assumed that CTRs drop for \emph{all} ads for each unsuccessful advertisement shown then, in the long run, this puts more emphasis on showing ads with high CTR, regardless of how much their advertisers can be charged for a single click. Although different objective functions for the search engine have been considered \cite{AMT07}, the setting of profit maximization with trust decay has not been studied and we deem this an interesting area for future work.

Finally, there is previous work that is relevant on a more technical level. In particular solution concepts such as the Core \cite{G59,M97} or the Shapley-Value \cite{S53,M97} have been studied extensively before. The exact connection to this group of work will be made clear in the sections with our technical contributions.

\subsection{Classification of Advertising Schemes}\label{sec:classification}
One could argue that a recommendation is, ultimately, just an advertisement and that an advertisement is just a recommendation. To highlight the differences between different kinds of advertisement in general, we present a simple classification scheme.
\begin{enumerate}
\setlength{\itemsep}{0pt}
%\item[$\bullet$] The {\em position} of the seller is different: If I want to market a new book A I have many options
%                 (newspaper, radio, TV, ...), but if I recommend you to buy book A at price p then you might 
%                 have only one option (namely, book store B where book A is currently for sale).  
                 % less competition / monopoly or oligopoly
\item[$\bullet$] Addressing: Personal vs. general. A recommendation is per se more personal than an advertisement and should be adapted to reflect the individual needs and interests of the potential buyer. Classic advertisement is not personalized and uses the same ``message'' for everyone.
\item[$\bullet$] Trust: High vs. low. A recommendation should come from someone the potential buyer trusts and feels loyal or close to. This can be a personal friend or maybe a well-respected blogger. In classic advertisement the information source is viewed as less reputable, though advertisers try to use trusted icons for their purposes.
                 % footnote: buy a book != read a book ??
                 % conversion rate
\item[$\bullet$] Intention: Altruistic vs. commercial. The intention of a recommendation by a friend is generally not commercial. She might not get reimbursed at all but she still recommends something as she believes you would profit from it. In ordinary advertising the reason for the act of advertising itself is a commercial one. %Web search engines do not display ads because they believe the web user really needs the corresponding products but they do so out of a commercial interest.
\end{enumerate}
The first part of this paper (Section~\ref{sec:pricing}) considers the setting of personal, highly trusted and altruistic recommendations. The second part (Section~\ref{sec:recommending-strategically}) then investigates the case of still personal, but commercial recommendations with a decaying amount of trust involved. To demonstrate the general applicability of this schema, we use it to classify a number of different advertising scenarios.
\begin{enumerate}
\setlength{\itemsep}{0pt}
\item Billboards. A chain of pizza restaurants puts up billboards all over the country, without targeting any specific group. Addressing: general, trust: low, intention: commercial.
%\item Television commercial. A company for home trainers runs tv commercials during half-time of a soccer match, hoping to have a fitness-oriented audience. Addressing: mildly personal, trust: low, intention: commercial.
\item Sponsored search. A web search engine shows targeted sponsored results in addition to ``organic'' web search results, trying to match the searcher's intent. Addressing: personal, trust: low, intention: commercial.
\item Testimonial. You liked a book and you write a testimonial on Amazon to convince other unknown readers to read it, too. Addressing: general, trust: high, intention: altruistic.
\item Direct recommendation. A friend asks you for advice on which laptop to buy and you recommend the model which you believe is best for her. Addressing: personal, trust: high, intention: altruistic.
%\item Public health campaign. A country's health department runs a campaign to warn of the risks of smoking. Addressing: general, trust: medium, intention: altruistic.
%\item Pyramid scheme. A friend tries to convince you to join a get-rich-fast scheme where you have to recruit additional members to make money. Addressing: personal, trust: high, intention: commercial.
\end{enumerate}

Of course, there are lots of other important differences, e.g.~concerning the conversion rates, but we view these differences as consequences of the ``axiomatic'' differences above and we assume that a personalized, altruistic ``advertisement'' from a highly trusted source will always have a higher conversion rate than a general, commercial ``recommendation'' from a disreputable source.

% THIS IS THE OLD FIGURE
% \begin{figure}[h!]
% \centering
% \epsfig{file=classification.eps,width=6.5cm}
% \caption{Classification of the advertising schemes discussed in the text, using the two dimensions ``Adressing'' and ``Intention''.}\label{fig:1}
% \end{figure}
%
% THIS IS THE NEW FIGURE
%\begin{figure}[h!]
%\centering
%\epsfig{file=star-coord.eps,width=6.5cm}
%\caption{Add triangle graph here. Classification of the advertising schemes discussed in the text. Direct recommendation is the most successful advertising medium as it dominates all other schemes in all dimensions.}\label{fig:1}
%\end{figure}
%
\begin{figure}[h!]
\centering
\epsfig{file=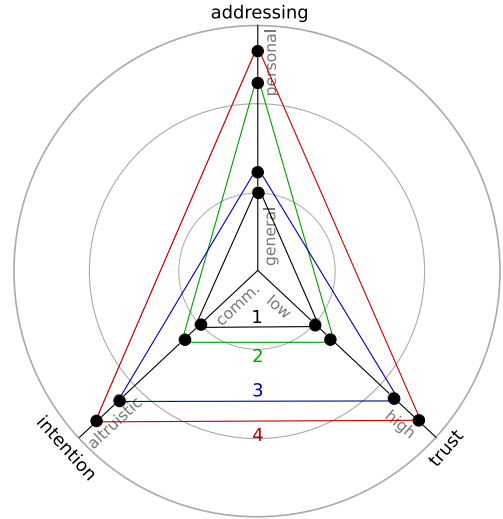,width=6.5cm}
%\vspace{-3mm}
\caption{Visualization of the four advertising schemes discussed in the text. Direct recommendation (\#4) is the most successful advertising medium as it dominates all other schemes in all dimensions.}\label{fig:1}
\end{figure}
\subsection{Our Contributions and Outline}
To the best of our knowledge there has been no work focusing on either (i) the pricing of recommendations (our Section~\ref{sec:pricing}) or (ii) the strategic behavior of recommenders in a setting with decaying trust (our Section~\ref{sec:recommending-strategically}). We view the introduction of these problems as one of our contributions.

As far as the pricing of recommendations is concerned we prove that there can be no pricing mechanism that is ``truthful'' with respect to the seller (Section~\ref{sec:truthful}). This shows that the seller can always pretend to profit less from the recommendations than she actually does to get a larger piece of the pie. We then apply solution concepts from coalitional game theory, namely the Core, the Shapley value, and the Nash Bargaining Solution, to determine provably ``fair'' prices. 
For the Core we find that it typically contains all ``individual rational'' payoff vectors, including the payoff vector where the seller gets everything and the recommenders get nothing. On the one hand, this demonstrates the weakness of the recommenders: They cannot form a coalition with non-zero value without the seller. On the other hand, it shows that the Core is essentially useless for deciding how to distribute the ``extra profit'' the seller can expect from being recommended among the recommenders (Sections~\ref{sec:one-core} and \ref{sec:many-core}). 
For the Shapley value we find that it not only defines unique prices, but that these prices are also ``fair'' in a very intuitive way: The price of a recommendation should be proportional to the ``extra profit'' the seller can expect from it (Section~\ref{sec:one-shapley} and~\ref{sec:many-shapley}). 
For the Nash Bargaining Solution we find that it yields ``fair'' prices, namely those obtained by the Shapley value, only if there is a single recommender. Otherwise, especially in situations where the recommenders do {\em not} contribute equally to the ``extra profit'' of the seller, it may lead to ``unfair'' prices (Section~\ref{sec:one-nash} and \ref{sec:many-nash}).
Finally, we also consider the case where each recommendation consists of one or more ``purchase arguments''. Here the ordinary Shapley value is no longer the method of choice, as withholding arguments might be beneficial for the recommenders. We show how the anonymity-proof Shapley value from~\cite{OCS08} can be applied to overcome this problem (Section~\ref{sec:many-proofshapley}).

%For the setting where different recommenders have different purchase arguments we give an example to show that, when the ordinary Shapley Value is used, withholding arguments can lead to a higher reward for recommenders (Example \ref{XXXX}) and we then show how the strategy-proof Shapley Value  can be applied to .

In the second part on the strategic behavior of profit maximizing recommenders we first show that, not surprisingly, the total expected profit of the recommender is bounded when the recommendee can only lose and does not regain trust (Section~\ref{sec:wo_reset_wo_recovery}). Then we prove that the total expected profit is still bounded over an infinite (!) sequence of recommendations, even when trust is reset to an initial level on each successful recommendation (Section:~\ref{sec:with_reset_wo_recovery}). Finally, we show that when trust is regained incrementally when no recommendations are made, the recommender's optimal total expected profit is unbounded in the long run and that she can recommend both too aggressively and too passively (Section~\ref{sec:with_reset_with_recovery}). These results are also applicable to the phenomenon of ``banner blindness''.
%
%%%%%%%%%%%%%%%%%%%%%%%%
% PART I: FAIR PRICING %
%%%%%%%%%%%%%%%%%%%%%%%%
%
%\section{Part I: The Pricing of Recommendations}\label{sec:pricing}
\section{The Pricing of Recommendations}\label{sec:pricing}
%
%\label{sec:preliminaries}
%
%A {\em coalitional game with transferable payoff} consists of (1) a finite set $N$ (the set of {\em players}) and (2) a function $v$ that associates with every non-empty subset $S$ of $N$ (a {\em coalition}) a real number $v(S)$ (the {\em worth} of $S$). For each coalition $S$ the number $v(S)$ is the total payoff that is available for division among the members of $S$. That is, the set of joint actions that the coalition $S$ can take consists of {\em all} possible divisions of $v(S)$ among the members of $S.$ 
%
%A coalitional game $\langle N,v \rangle$ is {\em cohesive} if $v(N) \geq \sum_k v(S_k)$ for every partition $\{S_1, \dots, S_k\}$ of $N.$ That is, the worth of the coalition $N$ of all players is at least as large as the sum of the worths of the members of any partition of $N.$ We restrict our attention to cohesive coalitional games.
%
%\subsection{Model and Scenarios}\label{sec:model}
%
We model the pricing of recommendations problem as a {\em coalitional game with transferable payoff} $\langle N,v \rangle$, where $N$ is a finite set (the set of {\em players}) and $v$ is a function that associates with every non-empty subset $S$ of $N$ (a {\em coalition}) a real number $v(S)$ (the {\em worth} of $S$). We use $s$ to denote the {\em seller}, who is paying for recommendations, and $r_i$ to denote the $i${-th} recommender. There is exactly one {\em product} for sale.\footnote{Note that this does {\em not} restrict the generality of our model. It rather says that each recommendation is for a distinct entity that we refer to as a product.} For each coalition $S$ the number $v(S)$ is the total payoff that is available for division among the members of $S$. We use $\delta \ge 0$ to denote the seller's {\em margin} or {\em gain} from selling the product, i.e. the sales price minus the production cost, and distinguish three scenarios for $v$:

%In our setting, the set of players is $N = \{s,r_1, \dots, r_n\}$, where $s$ is the seller and $r_i$ is the $i${-th} recommender. We use $\delta \ge 0$ to denote the gain generated through the sale of a (single) product. 
%% Typically, this gain will be the difference between the price of the product and the cost of producing it. We define $v(S)$ to be the expected gain of the coalition $S.$
% We consider the following scenarios:
%% We assume that if no recommendation was given, i.e. $S = \{s\}$, then the product is sold with probability $p \in [0,1]$ and if each of the recommenders $R \subseteq \{r_1, .., r_n\}$ recommends the product, i.e. $S = \{s\} \cup R$, then it is sold with probability $p+f(R \cup \{s\})$, where $f: 2^N \rightarrow [0,1-p]$ is an arbitrary function. We refer to this model as $\langle N,v \rangle$ (General). Two remarkable special cases of this model are:
%
% We consider three scenarios.
%
\begin{enumerate}
\setlength{\itemsep}{0pt}
\item[$\bullet$] {\em General.} Without any recommendation the product is sold with probability $p \in [0,1].$ If the recommenders $R \subseteq \{r_1, .., r_n\}$ recommend the product, then the probability that the product is sold is $p+f(\{s\} \cup R)$, where $f: 2^N \rightarrow [0,1-p]$ is an arbitrary function with $f(\{s\}) = 0$.% The following two scenarios are special cases of this case.
\end{enumerate}
The following two scenarios are special cases of {\em General}.
\begin{enumerate}
\setlength{\itemsep}{0pt}
\item[$\bullet$] {\em Linear.} Without any recommendation the product is sold with probability $p \in [0,1].$ The recommendation of the $i${-th} recommender increases this probability by $q_i \in [0, 1 - \sum_{j \neq i} q_j].$ The joint effect of more than one recommendation is the sum of the effect of the individual recommendations. Formally, if the recommenders $R \subseteq \{r_1, .., r_n\}$ recommend the product, then the probability is $p+\sum_{i: r_i \in R } q_i.$ 
\item[$\bullet$] {\em Threshold.} If less than $k$ recommenders recommend the product, then the product is sold with probability $p \in [0,1].$ If at least $k$ recommenders recommend the product, then it is sold with probability $p+q$ where $q \in [0,1-p].$ 
\end{enumerate} 
We refer to these scenarios as $\langle N, v \rangle$ (General), $\langle N, v \rangle$ (Linear), and $\langle N, v \rangle$ (Threshold). The following table gives the worth $v(S)$ of all $S \subseteq N = \{s, r_1, r_2\}$ for all three scenarios.
% Table \ref{tab:worths} gives $v(S)$ for all $S \subseteq N = \{s, r_1, r_2\}$ (i.e. $n = 2$) and all three scenarios.
%
% TABLE: EXAMPLES
\begin{table}[ht!] %changed h! to ht! to turn off warning
\centering
%\caption{Examples.}
%\setlength{colsep}{1mm}
\setlength{\tabcolsep}{1.5mm}
\begin{tabular}{|c|ccc|}
\hline
$S$         & Linear & Threshold & General \\ 
\hline
$\emptyset$        &0                        &0                      &0                \\
$\{s\}$            &$p \delta$               &$p \delta$             &$p \delta$       \\
$\{r_1\}$          &0                        &0                      &0                \\
$\{r_2\}$          &0                        &0                      &0                \\
$\{s, r_1\}$       &$(p + q_1) \delta$       &$p \delta$             &$(p+f(s,r_1))\delta$   \\
$\{s, r_2\}$       &$(p + q_2) \delta$       &$p \delta$             &$(p+f(s,r_2))\delta$   \\
$\{r_1, r_2\}$     &0                        &0                      &0                \\
$\{s, r_1, r_2 \}$ &$(p + q_1 + q_2) \delta$ &$(p + q) \delta$       &$(p+f(s,r_1,r_2))\delta$  \\
\hline\end{tabular}
%\vspace{-3mm}
\caption{Worths $v(S)$ of all possible coalitions $S$ for one seller $s$ and two recommenders $r_1$ and $r_2$ for our three different models.
% These values are used to compute the Shapley value (Theorems~\ref{the:shapley-one} and \ref{the:shapley-many}).
}\label{tab:worths}
\end{table}

% Our goal is to find a ``fair'' payoff vector $(x_s,x_{r_1}, \dots x_{r_n})$, where $x_s$ denotes the expected payoff to the seller and $x_{r_i}$ denotes the expected payoff to the $i${-th} recommender, i.e. the amount she is paid by the seller. Since such a payoff vector gives only expected payoffs, it can be translated into a {\em pricing scheme} in two ways:
%Our goal is to find a ``fair'' payoff vector $(x_s,x_{r_1}, \dots x_{r_n})$, where $x_s$ denotes the expected payoff to the seller and $x_{r_i}$ denotes the expected payoff to the $i${-th} recommender. If the seller $s$ is recommended by the recommenders $r_i \in R \subseteq N \setminus \{s\}$, then the worth of the coalition $\{s\} \cup R$ is $v(\{s\} \cup R) = (p + f(\{s\} \cup R)) \cdot \delta.$ And whenever the total payoff to the recommenders $\sum_{r_i \in R} x_{r_i}$ is no higher than $v(\{s\} \cup R)$, then the payoff vector $(x_s,x_{r_1}, \dots x_{r_n})$ can be translated into {\em prices}, i.e. payments from the seller to the recommenders, as follows:
Our goal is to find a payoff vector $(x_s,x_{r_1}, \dots x_{r_n})$, where $x_s$ denotes the expected payoff to the seller and $x_{r_i}$ denotes the expected payoff to the $i${-th} recommender. Suppose that the seller $s$ is recommended by all recommenders $r_i \in N \setminus \{s\}$, then the worth of this coalition is $v(N) = (p + f(N)) \cdot \delta.$ We say that the payoff vector $(x_s,x_{r_1}, \dots x_{r_n})$ is {\em feasible} if $x_s + \sum_i x_{r_i} = v(N)$. A feasible payoff vector, which prescribes the {\em expected} payoff to each player, can be translated into {\em prices}, i.e. payments from the seller to the recommenders, as follows:
\begin{enumerate}
\setlength{\itemsep}{0pt}
\item {\em Pay-per-Recommendation}: The recommender gets paid by the seller for every recommendation; successful or not. That is, on every recommendation the seller $s$ pays the $i${-th} recommender $r_i$ the money equivalent of $x_{r_i}.$

\item {\em Pay-per-Sale}: The recommender gets paid by the seller for successful recommendations only. That is, on every successful recommendation the seller $s$ pays the $i${-th} recommender $r_i$ the money equivalent of $1/(p+f(N)) \cdot x_{r_i}$. %, where $R$ is the set of recommenders.
\end{enumerate}
%
%In practice, the Pay-per-Sale approach might be preferable as, on a successful recommendation, one could reasonably assume $p+f(\{s\} \cup R)=1$, sidestepping the problem of estimating $f(\{s\} \cup R)$ with very little or no data. Note that the prior probability $p$ is easier to estimate using the seller's sales record and click-through or conversion-rates.
In practice, the Pay-per-Sale approach might be preferable as, on a successful recommendation, one could reasonably assume $p+f(N)=1$, sidestepping the problem of estimating $f(N)$ with very little or no data. Note that the prior probability $p$ is easier to estimate using the seller's sales record and click-through or conversion-rates.

\subsection{Impossibility Result}\label{sec:truthful}
Ideally, the payoff vector $(x_s, x_{r_1}, \dots, x_{r_n})$ computed by whatever mechanism should give the seller $s$, who holds the private information on $p$, $f$, and $\delta$, the incentive to reveal her information {\em truthfully}. Formally, we want that for all $p'$, $f'$, and $\delta':$ $x_s(p,f,\delta) \ge x_{s}(p',f',\delta')$, where $x_s(p,f,\delta) = (p+f(N)) \cdot \delta - \sum_{i} x_{r_i}(p,f,\delta)$ and $x_{s}(p',f',\delta') = (p+f(N)) \cdot \delta - \sum_{i} x_{r_i}(p',f',\delta').$ Unfortunately, as the following theorem shows, the only truthful payoff vector has $\sum_{i} x_{r_i} = 0$. That is, the seller gets everything and the recommenders get nothing.  
\begin{theorem}\label{the:truthful}
There can be no truthful payoff vector $(x_s,$ $x_{r_1}$, $\dots, x_{r_n})$ that has $\sum_{i} x_{r_i} \neq 0$ and ensures participation of the seller $s$ and the recommenders $r_1$ to $r_n.$ 
\end{theorem}
\begin{proof}
To ensure participation for the seller, we must have $x_s(p,f,\delta) = (p+f(N)) \cdot \delta 
- \sum_i x_{r_i}(p,f,\delta) \ge 0$ for all $p$, $f$, and $\delta.$ To ensure participation 
for the recommenders $r_1$ to $r_n$ we must have $\sum_i x_{r_i}(p,f,\delta) \ge 0$ for all $p$, 
$f$, and $\delta$. Now suppose $(x_s,$ $x_{r_1}$, $\dots, x_{r_n})$ with $\sum_{i} x_{r_i} \neq 0$
was truthful. It follows that $x_s(p,f,\delta) \ge x_s(p',f',\delta')$ for all $p'$, $f'$, and
$\delta'$, i.e. $(p+f(N))\cdot \delta - \sum_i x_{r_i}(p,f,\delta) \ge (p+f(N)) \cdot \delta - 
\sum_i x_{r_i} (p',f',\delta').$ And hence, $\sum_i x_{r_i} (p',f',\delta') \ge \sum_i x_{r_i}
(p,f,\delta).$ But since $\sum_{i} x_{r_i}(p,f,\delta) > 0$ there must be $p'$, $f'$, $\delta'$ 
such that$(p'+f'(N)) \cdot \delta' < \sum_i x_{r_i}(p,f,\delta)$ with $\sum_i x_{r_i}(p',f',\delta') 
\le (p'+f'(N)) \cdot \delta' < \sum_i x_{r_i}(p,f,\delta).$ Contradiction!
\end{proof}

%\begin{proof}
%To ensure participation for the seller, we must have $u_s(\delta, p(\delta)) \ge 0$ for 
%all $\delta$, i.e. $p(\delta) \leq \delta$ for all $\delta.$ To ensure participation for the 
%recommender we must have $p(\delta) \ge 0$ for all $\delta$. Now suppose $p(\delta) \neq 0$ was 
%truthful. It follows that $u_s(\delta,p(\delta)) \geq u_s(\delta,p(\delta'))$ for all $\delta'$, 
%i.e. $p(\delta') \geq p(\delta)$ for all $\delta'.$ But since $p(\delta) > 0$ there must be
%a $\delta' < p(\delta)$ for which $p(\delta') \le \delta' < p(\delta).$ Contradiction! 
%\end{proof}

%%%%%%%%%%%%%%%%%%%
% ONE RECOMMENDER %
%%%%%%%%%%%%%%%%%%%
%
\subsection{One Recommender}\label{sec:one}
We begin by studying the problem of finding ``fair'' prices in the setting $N = \{s,r\}$, i.e., there is only one seller and one recommender. In this setting the games $\langle N, v \rangle$ (Linear) and $\langle N, v \rangle$ (Threshold) are equivalent. We discuss the solution concepts {\em Core}, {\em Shapley value}, and {\em Nash Bargaining Solution}. For a more detailed discussion of these solution concepts see~\cite{M97,OR94}. 
% 
% SUBSECTION: THE CORE
%
\subsubsection{The Core}\label{sec:one-core}
%
% The {\em Core}~\cite{G59} of a coalitional game is an outcome of cooperation among all players where no coalition of players can lead to a strict benefit to all its members for breaking away from the grand coalition. Intuitively, the core of a game corresponds to situations where it is possible to sustain cooperation among all players in an economically stable manner. 

The {\em Core}~\cite{G59} of a coalitional game is an outcome of cooperation among all players where no coalition of players can obtain higher payoffs for all of its members. A payoff vector in the Core is ``fair'' in the sense that no subset of players can justifiably argue that they are paid to little, as they are unable to achieve higher payoffs on their own.

More formally, the Core of the game $\langle N,v \rangle$ is the set of {\em feasible} payoff vectors  $(x_i)_{i \in N}$ for which there is no coalition $S \subseteq N$ and $S$-{\em feasible} payoff vector $(y_i)_{i \in N}$ such that $y_i > x_i$ for all $i \in S.$ Recall that a payoff vector $(x_i)_{i \in N}$ is {\em feasible} if $\sum_{i \in N} x_i = v(N)$. It is $S$-{\em feasible} if $\sum_{i \in S} x_i = V(S)$.

%A vector $(x_i)_{i \in S}$ of real numbers is an $S$-{\em feasible payoff vector} if $x(S) = \sum_{i \in S} x_i = v(S).$ We refer to a $N$-feasible payoff vector as {\em feasible payoff vector}. A feasible payoff vector $(x_i)_{i \in N}$ {\em cannot be improved upon}, if there is no coalition $S \subseteq N$ and $S$-feasible payoff vector $(y_i)_{i \in S}$ for which $y_i > x_i$ for all $i \in S.$ 

The Core can be shown to be non-empty by means of the Bondareva-Shapley Theorem~\cite{B63, S67}, which states that a game has a non-empty core if and only if it is balanced. A game $\langle N, v \rangle$ is {\em balanced} if for every balanced collections of weights $(\lambda_S)_{S \subseteq N}$: $\sum_{S} \lambda_S \cdot v(S) \leq v(N).$ A {\em balanced collection of weights} $(\lambda_S)_{S \subseteq N}$ is a collection of numbers $\lambda_S \in [0,1]$ (one for each coalition $S \subseteq N$) such that for all $i$: $\sum_{S \subseteq N: i \in S} \lambda_S = 1.$

% One interpretation of a balanced game is: Each player has one unit of time, which she must distribute among all the coalitions of which she is a member. In order for a coalition $S$ to be active for the fraction of time $\lambda_S$, all its members must be active in $S$ for this fraction of time, in which case the coalition yields the payoff $\lambda_S \cdot v(S).$ In this interpretation the condition that the collection of weights be balanced is a feasibility condition on the players' allocation of time that yield the players no more than $v(N).$ 

% Denote by $C$ the set of all coalitions, for any coalition $S$ denote by $\mathbb{R}^n$ the $|S|$-dimensional Euclidean space in which the dimensions are indexed by the members of $S$, and denote by $1_S \in \mathbb{R}^N$ the characteristic vector of $S$ given by $(1_S)_i = 1$ if $i \in S$ and $(1_S)_i = 0$ otherwise. A collection $(\lambda_S)_{S \in C}$ of numbers in $[0,1]$ is a {\em balanced collection of weights} if for every player $i$ the sum of $\lambda_S$ over all the coalitions that contain $i$ is $1$: $\sum_{S \in C} \lambda_S 1_S = 1_N.$ A game $\langle N,v \rangle$ is {\em balanced} if $\sum_{S \in C} \lambda_S v(S) \leq v(N)$ for every balanced collection of weights. A coalitional game with transferable payoff has a non-empty core if and only if it is balanced. This result is known as the Bondareva-Shapley Theorem.

\begin{theorem}\label{the:one-core}
The game $\langle \{s,r\},v \rangle$ (General) has a non-empty core. 
%... if and only if:
% \begin{align*}
% &f(s,r) \cdot \delta \geq 0.
% \end{align*}
\end{theorem}
\begin{proof}
Let $x \in [0,1].$ All balanced collections of weights $(\lambda_S)_{S \subseteq N}$ are of the form $\lambda_{S} = x$ for $S = \{s\}, \{r\}$ and $\lambda_{S} = 1-x$ for $S = \{s,r\}.$ By the Bondareva-Shapley Theorem, the Core is non-empty if and only if for all values $x \in [0,1]:$
\begin{align*}
x \cdot (v(\{s\}) + v(\{r\})) + (1-x) v(\{s,r\}) & \leq v(\{s,r\}).
\end{align*}
For $x = 0$ this is trivially true. Next we we analyze the case $x > 0.$ Since $v(\{r\}) = 0$ and $v(\{s,r\}) - v(\{s\}) = f(s,r) \cdot \delta$, 
\begin{align*}
                & \ x \cdot v(\{s\}) - x \cdot v(\{s,r\}) &&\leq 0\\
\Leftrightarrow & \ v(\{s,r\}) - v(\{s\}) &&\geq 0\\
\Leftrightarrow & \ f(\{s,r\}) \cdot \delta &&\geq 0.
\end{align*}
Since $f(\{s,r\}) \ge 0$ and $\delta \ge 0$ this is always true.
\end{proof}
Recall that the games $\langle N, v \rangle$ (Linear) and $\langle N, v \rangle$ (Threshold) are special cases of the game $\langle N, v \rangle$ (General) and so Theorem~\ref{the:one-core} also shows non-emptiness of the Core for these games. Next we give necessary and sufficient conditions for a payoff vector $(x_s, x_{r_1}, \dots, x_{r_n})$ to be in the Core.
%
% For the game $\langle N, v \rangle$ (Linear) this means that the core is non-empty if and only if:
% \begin{align*}
% &q \cdot \delta \geq 0.
% \end{align*}
%
% The following theorem is a direct consequence of the definition of the Core.
%
\begin{theorem}
The payoff vector $((p+f(\{s,r\}))\delta - x, x)$ is in the Core of the game $\langle \{s,r\}, v\rangle$ (General) if and only if:
\begin{align*}
&0 \le x \leq f(\{s,r\}) \cdot \delta.
\end{align*}
\end{theorem}
\begin{proof}
Let $x_s = (p+f(\{s,r\}))\delta - x$ and let $x_r = x.$ If the vector $(x_s,x_r)$ is in the Core of the game $\langle \{s,r\}, v\rangle$ (General), then there exists no coalition $S \subseteq N = \{s,r\}$ and an $S${-feasible} payoff vector $y=(y_s,y_r)$ such that $y_i > x_i$ for all $i \in S.$ That is, for all $S \subseteq N$ and $S${-feasible} payoff vectors $y=(y_s,y_r)$ we have that $y_i \le x_i$ for all $i \in S.$ For $S = \{s,r\}$ this means that $y_s + y_r = v(\{s,r\}) \le x_s + x_r = v(\{s,r\})$ (which is trivially true). For $S = \{s\}$ this means that $y_s = v(\{s\}) = p \cdot \delta \le x_s = (p + f(\{s,r\})) \cdot \delta - x$, i.e. $x \le f(\{s,r\}) \cdot \delta.$ For $S = \{r\}$ this means that $y_r = v(\{r\}) = 0 \le x_r = x$, i.e. $x \ge 0$. That is, $0 \le x \le f(\{s,r\}) \cdot \delta.$

For the reverse direction assume by contradiction that $0 \le x \le f(\{s,r\}) \cdot \delta$ but that $(x_s,r_r)$ is {\em not} in the Core, i.e. there exists a coalition $S \subseteq N = \{s,r\}$ and a $S${-feasible} payoff vector $y=(y_s,y_r)$ such that $y_i > x_i$ for all $i \in S.$ We cannot have $S = \{s,r\}$ as then $y_s+y_r = v(\{s,r\}) > x_s + x_r = v(\{s,r\})$, which gives a contradiction. But if $S = \{s\}$, then $y_s = v(\{s\}) = p \cdot \delta > x_s = (p + f(\{s,r\}))\cdot \delta - x$, i.e. $x > f(\{s,r\}) \cdot \delta$, which gives a contradiction. Finally, if $S = \{r\}$, then $y_r = v(\{r\}) = 0 > x_r = x$, i.e. $x < 0$, which also gives a contradiction.
\end{proof}
For the games $\langle \{s,r\}, v\rangle$ (Linear) and $\langle \{s,r\}, v \rangle$ (Threshold) this means that $((p+q)\delta - x, x)$ is in the core if and only if:
\begin{align*}
&0 \le x \leq q \cdot \delta.
\end{align*}
This implies that any ``feasible'' payoff vector is in the Core. The only restriction on the payoff vector is that the payoff to the recommender be non-negative and in expectation no higher than the ``extra profit'' the seller can expect from the recommendation. In particular, a payoff vector that gives everything to the seller and nothing to the recommender would be in the Core. This demonstrates the weakness of the recommenders: They cannot form a coalition with non-zero value without the seller.
%
%% SUBSECTION: SHAPLEY VALUE %%
%
\subsubsection{Shapley Value}\label{sec:one-shapley}
One problem with the Core is that it does {\em not} assign a unique payoff vector to a game. This makes it necessary to have another criterion for choosing a payoff vector. The {\em Shapley value}~\cite{S53} is a solution concept that assigns to each game a unique, provably fair payoff vector. In general, a {\em value} $\phi: v \rightarrow \mathbb{R}^{n+1}$ maps each game $\langle N,v \rangle$ to a unique vector $\phi(v)$; the $i${-th} entry $\phi_i(v)$ of this vector being the {\em expected payoff} to player $i.$ The Shapley value is the unique value satisfying the following axioms:
\begin{enumerate}
\setlength{\itemsep}{0pt}
\item {\em Symmetry:} 
If player $i$ and $j$ are interchangeable, then $\phi_i(v) = \phi_j(v).$ Formally, if for every $S \subseteq N$ s.t. $i \in S$, $j \not\in S$: $v( (S \setminus \{i\}) \cup \{j\}) = v(S)$, then $\phi_i(v) = \phi_j(v).$
%If $i$ and $j$ are two players, and $w$ acts just like $v$ except that the roles of $i$ and $j$ have been exchanged, then $\phi_i(v) = \phi_j(w).$
\item {\em Dummy:} If player $i$'s contribution to any coalition $S$ is zero, then $\phi_i(v) = 0.$ Formally, if for every $S \subseteq N \setminus \{i\}$: $v(S \cup \{i\}) = v(S)$, then $\phi_i(v) = 0.$
%A player whose marginal contribution to {\em any} coalition is zero should have a value of zero: 
\item {\em Additivity:} Player $i$'s entry $\phi_i(v)$ should be additive in $v.$ Formally, if $\langle N, v + w\rangle$ is derived from $\langle N,v \rangle$ and $\langle N,w \rangle$ by defining $(v+w)(S) = v(S) + w(S)$ for all $S \subseteq N$, then $\phi_i(v+w) = \phi_i(v) + \phi_i(w)$ for all $i \in N$. 
\end{enumerate}
These axioms can be interpreted as formalizing a notion of ``fairness'', that postulates that the expected payoff to player $i$ be proportional to player $i$'s contribution to the outcome of the game. For an analysis along these lines see~\cite{M91}. 
\begin{definition} \label{def:shapley}
The Shapley value $\phi(v) = (\phi_1(v), \dots, \phi_N(v))$ of the game $\langle N, v\rangle$ is defined as follows:
\begin{align*}
\phi_i(v) &= \sum_{S \subseteq N \setminus \{ i \}} \frac{|S|!(|N|-1-|S|)!}{|N|!} \cdot (v(S \cup \{ i\}) - v(S)). 
\end{align*}
\end{definition}
One interpretation of this is: Suppose that all the players are arranged in some order, all orders being equally likely, then $\phi_i(v)$ is the {\em expected marginal contribution} of player $i$ to the set of players who precede her.
\begin{theorem}\label{the:shapley-one}
Consider the game $\langle \{s,r\}, v \rangle$ (General). The Shapley value $\phi(v) = (\phi_s(v), \phi_{r}(v))$ is given by
\begin{align*}
\phi_s(v)   = p \delta + \frac{1}{2} f(\{s,r\}) \delta \mbox{\hspace{2mm} and \hspace{2mm}} \phi_{r}(v) = \frac{1}{2} f(\{s,r\}) \delta.
\end{align*}
%\begin{align*}
%&\phi_s(v)   = p \delta + \frac{1}{2} f(\{s,r\}) \delta\\
%&\phi_{r}(v) = \frac{1}{2} f(\{s,r\}) \delta.
%\end{align*}
\end{theorem}
\begin{proof}
The claim follows from the definition of the game $\langle \{s,r\}, v \rangle$ (General) and Definition~\ref{def:shapley}. Note that the worths $v(S)$ of all $S \subseteq N = \{s,r\}$ can be read from Table~\ref{tab:worths} by treating $r_1$ as $r$ and ignoring rows containing $r_2.$ 
\end{proof}
For the game $\langle \{s,r\}, v \rangle$ (Linear, Threshold) this means that
\begin{align*}
\phi_s(v)     = p \delta + \frac{1}{2} q \delta  \mbox{\hspace{2mm} and \hspace{2mm}}\phi_{r}(v) = \frac{1}{2} q \delta.
\end{align*}
%\begin{align*}
%&\phi_s(v)     = p \delta + \frac{1}{2} q \delta && \mbox{and}\\
%&\phi_{r}(v) = \frac{1}{2} q \delta.
%\end{align*}
%
This shows that the payoff to the recommender should be proportional to her contribution to the seller's expected ``extra profit''. In particular, it shows that the recommender's payoff should be {\em linear} in her contribution to the purchase probability, i.e. $f(\{s,r\})$, and also in the seller's margin or gain $\delta.$ This is consistent with ``real life'' pricing schemes that redeem the recommender with a certain percentage of the sales price~\cite{LAH06}, assuming that for a given product family the margin is proportional to the sales price.

%if we follow the Shapley value, then the price of a recommendation should be $\frac{1}{2} q \delta.$ This is interesting as it shows that the seller $s$ and the recommender $r$ should receive the same share of the {\em extra value} $q \delta$ generated by $r$'s recommendation. 
%
% SUBSECTION: THE NASH BARGAINING SOLUTION
%
\subsubsection{The Nash Bargaining Solution}\label{sec:one-nash}
The last solution concept that we discuss in this section is the {\em Nash Bargaining Solution}~\cite{N50}.\footnote{The only connection between this solution concept and the concept of a {\em Nash equilibrium}~\cite{M97,OR94} is John F.~Nash.} 
%
% The basic idea here is to view the two-player game $\langle \{s,r\},v \rangle$ as a {\em bargaining problem} over a set $F$ of feasible payoff vectors $f = (f_1,f_2)$ and a dedicated payoff vector $d=(d_1,d_2)$; the payoff vector in case of a {\em disagreement}. A {\em solution} is a function $\phi: (F,d) \rightarrow F.$ The Nash Bargaining Solution is the unique solution satisfying the following axioms:
%
The basic idea here is to view the game $\langle N,v \rangle$ as a {\em bargaining problem} over a set $F$ of feasible payoff vectors $f = (f_0,\dots,f_n)$ and a dedicated payoff vector $d=(d_0,\dots,d_n)$; the payoff vector in case of a {\em disagreement}. A {\em solution} is a function $\phi: (F,d) \rightarrow F.$ The Nash Bargaining Solution is the unique solution satisfying the following axioms:
%
%\begin{enumerate}
%\setlength{\itemsep}{0pt}
%\item {\em Strong Efficiency.}
%If $f \ge \phi(F,d)$, then $f = \phi(F,d).$
%\item {\em Individual Rationality.}
%$\phi(F,d) \ge d.$
%\item {\em Scale Covariance.} If $F' = \{(\lambda_1 \cdot f_1 + \gamma_1, \lambda_2 \cdot f_2 + \gamma_2) \ | \ (f_1,f_2) \in F\}$ and $d' = (\lambda_1 \cdot d_1 + \gamma_1, \lambda_2 \cdot d_2 + \gamma_2)$, then $\phi(F',d') = (\lambda_1 \cdot \phi_1(F,d) + \gamma_1, \lambda_2 \cdot \phi_2(F,d) + \gamma_2).$ 
%\item {\em Independence of irrelevant alternatives.} 
%If $F' \subseteq F$ and $\phi(F,d) \in F'$, then $\phi(F',d) = \phi(F,d)$.
%\item {\em Symmetry.}
%If $F = \{ (f_1,f_2) \ | \ (f_2,f_1) \in F \}$ and $d_1 = d_2$, then $\phi_1(F,d) = \phi_2(F,d).$ 
%\end{enumerate}
%
\begin{enumerate}
\setlength{\itemsep}{0pt}
\item {\em Pareto Efficiency.} There is no $f \in F$ such that $f_i \ge \phi_i(F,d)$ for all $i \in N$ and $f_j > \phi_j(F,d)$ for at least one $j \in N.$ 
\item {\em Individual Rationality.}
For all $i \in N:$ $\phi_i(F,d) \ge d_i$.
\item {\em Scale Covariance.} If $F' = \{(\lambda_0 \cdot f_0 + \gamma_0, \dots, \lambda_n \cdot f_n + \gamma_n) \ | \ (f_0, \dots, f_n) \in F\}$ and $d' = (\lambda_0 \cdot d_0 + \gamma_0, \dots, \lambda_n \cdot d_n + \gamma_n)$, then $\phi(F',d') = (\lambda_0 \cdot \phi_0(F,d) + \gamma_0, \dots, \lambda_n \cdot \phi_n(F,d) + \gamma_n).$ 
\item {\em Independence of Irrelevant Alternatives.} 
If $F' \subseteq F$ and $\phi(F,d) \in F'$, then $\phi(F',d) = \phi(F,d)$.
\item {\em Symmetry.}
If $ (f_0, \dots, f_i, \dots, f_j, \dots, f_n) \in F$ implies $(f_0, \dots, f_j, \dots, f_i, \dots, f_n) \in F$ and $d_i = d_j$, then $\phi_i(F,d) = \phi_j(F,d).$ 
\end{enumerate}
The Nash Bargaining solution is ``fair'' in the sense that is {\em Pareto effcient}, i.e.~it is impossible to improve the payoff of one or more players without hurting that of others. One can show that it satisfies 
%The Nash Bargaining Solution $\phi: (F,d) \rightarrow F$ satisfies: 
$\phi(F,d) \in \mbox{argmax}_{f \in F} \prod_i (f_i - d_i)$~\cite{M97}. We use this to prove:
\begin{theorem}
For $\langle \{s,r\}, v \rangle$ (General) let $F = \{ (f_s,f_r) \ |$ $ \ f_s \ge 0, f_r \ge 0 \ \mbox{and} \ f_s + f_r = (p + f(\{s,r\})) \cdot \delta \}$ and $d = (d_s, d_r) = (p \cdot \delta,0)$. Then, 
\begin{align*}
\phi(F,d) = ((p+\frac{1}{2}f(\{s,r\}) \delta , \frac{1}{2}f(\{s,r\}) \delta ).
\end{align*}
\end{theorem}
For the game $\langle \{s,r\}, v \rangle$ (Linear, Threshold) this means that:
\begin{align*}
\phi(F,d) &= \left( \left(p+\frac{1}{2}q\right) \delta , \frac{1}{2}q\delta \right).
\end{align*}
This demonstrates that the Nash Bargaining Solution and the Shapley value coincide in our model. On the one hand, this is surprising as the axioms used to define the Shapley value and the Nash Bargaining Solution are quite different. On the other hand, this is intuitive as for two players there is only one non-trivial coalition to be considered for the Shapley value. So if this two-player coalition leads to a bigger payoff $P$ than the sum of its two non-cooperative ``atoms'', then the different symmetry axioms present in both solution concepts imply that this surplus should be divided 50-50. 
%This holds as long as the feasible payoff vectors are $F=\{(f_1,f_2) | f_1+f_2 \le P \}$, and so $(f_1,f_2) \in F \Leftrightarrow (f_2,f_1) \in F$, but  would stop to hold if, e.g. there were additional constraints on $F$ such as $f_1 \ge 2 \cdot f_2$. In such cases, the symmetry axiom of the Nash Bargaining Solution no longer applies.
%
This holds as long as the feasible payoff vectors are $F=\{(f_0,f_1) | f_0+f_1 \le P \}$, and so $(f_0,f_1) \in F \Leftrightarrow (f_1,f_0) \in F$, but would stop to hold if, e.g.~there were additional constraints on $F$ such as $f_0 \ge 2 \cdot f_1$. In such cases, the symmetry axiom of the Nash Bargaining Solution no longer applies.
%
%%%%%%%%%%%%%%%%%%%%%%%%%%%%%%
% SECTION: MANY RECOMMENDERS %
%%%%%%%%%%%%%%%%%%%%%%%%%%%%%%
%
\subsection{Many Recommenders}\label{sec:many}
Next we study the problem of finding ``fair'' prices in the more general setting $N = \{s, r_1, \dots, r_n\}$, i.e. there is one seller and $n \ge 1$ recommenders. Note that in this setting the games $\langle N,v \rangle$ (Linear) and $\langle N,v \rangle$ (Threshold) are no longer equivalent. 
% PD: REPLACED THE FOLLOWING SENTENCE(S)
% We focus on the solution concepts Core and Shapley value as the Nash Bargaining Solution is only defined for two-player games.
% As in the setting where $N = \{s,r\}$ we focus on the solution concepts Core and Shapley value.
As in the setting where $N = \{s,r\}$ we study the solution concepts Core, Shapley value, and Nash Bargaining Solution.
\subsubsection{The Core}\label{sec:many-core}
Recall that the Core of the game $\langle N, v \rangle$ comprises all feasible payoff vectors with which no coalition $S \subseteq N$ is ``unhappy'' meaning that the players in $S$ cannot break away to obtain a higher payoff on their own. For a formal definition of the Core (and related definitions) see Section~\ref{sec:one-core}. 
\begin{theorem}\label{the:many-core-1}
The game $\langle N,v \rangle$ (General) has a non-empty core iff for every balanced collections of weights $(\lambda_S)_{S \subseteq N}$:
\begin{align*}
&\sum_{S \subset N: s \in S} \left[ \lambda_{S} (f(N) - f(S)) \right] \geq 0.
\end{align*}
\end{theorem}
\begin{proof}
Let $(\lambda_S)_{S \subseteq N}$ be a balanced collection of weights. Since $v(S) = 0$ whenever $s \not\in S$, applying the Bondareva-Shapley Theorem~\cite{B63, S67} to the game $\langle N,v \rangle$ (General) gives:
\begin{align*}
&\sum_{S \subseteq N: s \in S} \left[ \lambda_S (p + f(S)) \right] \leq p + f(N).
\end{align*}
Since $(\lambda_S)_{S \subseteq N}$ is a balanced collections of weights, we have $\sum_{S \subseteq N: s \in S} \lambda_S = 1$ and $\lambda_N = 1 - \sum_{S \subset N: s \in S} \lambda_S.$ It follows that:
\begin{align*}
                & \sum_{S \subseteq N: s \in S} \lambda_S \cdot f(S) & \leq f(N).\\
\Leftrightarrow & \sum_{S \subset N: s \in S} [\lambda_S (f(S) - f(N))] + f(N) & \leq f(N)\\
\Leftrightarrow &\sum_{S \subset N: s \in S} \left[ \lambda_S (f(N) - f(S)) \right] &\geq 0. &&\qedhere
\end{align*}
\end{proof}
The condition given by Theorem~\ref{the:many-core-1} holds trivially for $\langle N,v \rangle$ (General) if $f(N) \ge f(S)$ for all $S \subseteq N$ since all the $\lambda_S$ values are non-negative. For the game $\langle N,v \rangle$ (Linear) and $\langle N,v \rangle$ (Threshold) this means that the core is {\em always} non-empty since $q_i \ge 0$ for all $i$ respectively $q > 0$. 
\begin{theorem}
Consider the game $\langle N, v\rangle$ (General). The payoff vector $(x_s = (p+f(N)) \cdot \delta - \sum_{i} x_{r_i}, x_{r_1}, \dots, x_{r_n})$ is in the Core if and only if for all $T \subseteq N$ s.t.~$s \not\in T$:
\begin{align*}
&0 \le \sum_{r_i \in T} x_{r_i} \leq (f(N) - f(N \setminus T)) \cdot \delta.
\end{align*}
\end{theorem}
\begin{proof}
Assume that the payoff vector $(x_s = (p+f(N)) \cdot \delta - \sum_{i} x_{r_i}, x_{r_1}, \dots, x_{r_n})$ is in the Core. Since $v(S) = 0$ for all coalitions $S \subseteq N$ such that $s \not\in S$, it follows that: 

\vspace{6pt}\noindent 1. For all $S \subseteq N$ such that $s \in S:$
\begin{align*}
 &\ (p + f(N)) \delta - \sum_{r_i \in N} x_i + \sum_{r_i \in S} x_i &&\geq v(S)\\
\Leftrightarrow &\ (p + f(N)) \delta  - \sum_{r_i \in N \setminus S} x_i &&\geq (p+f(S))\delta\\
\Leftrightarrow &\ (f(N) - f(S)) \delta &&\geq \sum_{r_i \in N \setminus S} x_i 
\end{align*}
\noindent 2. For all $S \subseteq N$ such that $s \not\in S$:
\begin{align*}
\sum_{r_i \in S} x_i \geq v(S) \ \Leftrightarrow \ \sum_{r_i \in S} x_i \geq 0.
\end{align*}
With $T = N \setminus S$ in 1.~and $T = S$ in 2.~it follows that $0 \le \sum_{r_i \in T} x_{r_i} \leq (f(N) - f(N \setminus T)) \cdot \delta$ for all $T \subseteq N$ s.t.~$s \not\in T$.

For the reverse direction assume by contradiction that $0 \le \sum_{r_i \in T} x_{r_i} \leq (f(N) - f(N \setminus T)) \cdot \delta$ for all $T \subseteq N$ s.t.~$s \not\in T$ but that $(x_s = (p+f(N)) \cdot \delta - \sum_i x_{r_i}, x_{r_1}, \dots, x_{r_n})$ is {\em not} in the Core, i.e. there exists a coalition $S \subseteq N$ and an $S${-feasible} payoff vector $y=(y_s,y_{r_1}, \dots, y_{r_n})$ in which $y_k > x_k$ for all players $k$ in $S.$ Since $y$ is $S${-feasible}, the total payoff $\sum_{k \in S} y_k$ to the players in $S$ must equal $v(S).$ Since $y_p > x_p$ for all players $p$ in $S$, we must have $\sum_{k \in S} y_k > \sum_{k \in S} x_k.$ Thus, $v(S) > \sum_{k \in S} x_k.$
%\begin{align*}
%v(S) &> \sum_{k \in S} x_k.
%\end{align*}

{\em Case 1:} If $s \in S$, since $v(S) = (p + f(S)) \cdot \delta$ and $\sum_{k \in S} x_k = (p+f(N))\delta - \sum_{r_i \in N \setminus S} x_{r_i}$, this means that: 
%$(p + f(S)) \cdot \delta > (p+f(N))\delta - \sum_{r_i \in N \setminus S} x_{r_i}$, i.e. $\sum_{r_i \in N \setminus S} x_{r_i} > (f(N) - f(S)) \cdot \delta.$
\begin{align*}
                  & (p + f(S)) \cdot \delta &&> (p+f(N))\delta - \sum_{r_i \in N \setminus S} x_{r_i}\\
\Leftrightarrow \ & \sum_{r_i \in N \setminus S} x_{r_i} &&> (f(N) - f(S)) \cdot \delta.
\end{align*}
With $T = N \setminus S$ this gives a contradiction to the fact that for all $T \subseteq N$ s.t.~$s \not\in T$: $\sum_{r_i \in T} x_{r_i} \leq (f(N) - f(N \setminus T)) \cdot \delta.$

{\em Case 2:} If $s \not\in S$, since $v(S) = 0$ and $\sum_{k \in S} x_k = \sum_{r_i \in S} x_{r_i}$, this means that: 
% $v(S) = 0 > \sum_{r_i \in S} x_{r_i}.$
\begin{align*}
v(S) = 0 &> \sum_{r_i \in S} x_{r_i}.
\end{align*}
With $T = S$ this gives a contradiction to the fact that for all $T \subseteq N$ s.t.~$s \not\in T$: $\sum_{r_i \in T} x_{r_i} \ge 0.$
\end{proof}

For the game $\langle N, v\rangle$ (Linear) this means that the payoff vector $(x_s = (p+\sum_i q_i) \cdot \delta - \sum_i x_{r_i}, x_{r_1}, \dots, x_{r_n})$ is in the Core if and only if for all $T \subseteq N$ s.t.~$s \not\in T$:
\begin{align*}
&0 \le \sum_{r_i \in T} x_{r_i} \leq \left( \sum_{r_i \in T} q_i \right) \cdot \delta.
\end{align*}

For the game $\langle N, v\rangle$ (Threshold) this means that the payoff vector $(x_s = (p+q)\delta - \sum_i x_{r_i}, x_{r_1}, \dots, x_{r_n})$ is in the Core if and only if for all $T \subseteq N$ s.t.~$s \not\in T$: 
\begin{align*}
&0 \le \sum_{r_i \in T} x_i \leq \begin{cases} 0 & \text{, if} \ |T| \le n-k\\ q \cdot \delta& \text{, if} \ |T| > n-k \end{cases}.
\end{align*}
This means that for $\langle N, v\rangle$ (Linear) and $\langle N, v\rangle$ (Threshold) with $k = n$ a certain payoff vector is in the core precisely if no coalition of recommenders receives more than their joint contribution to the seller's expected ``extra profit''. 
% PD: REPLACED THE FOLLOWING SENTENCE(S)
% For the game $\langle N, v\rangle$ (Threshold) with $k < n$ this shows that the Core is empty. To conclude, the Core either contains many payoff vectors (making it necessary to have another criterion for choosing a ``good'' one), or empty (making it essentially useless). 
For the game $\langle N, v \rangle$ (Threshold) with $k < n$ this
means that $x_{r_i} = 0$ for all $i$ (with $T = \{ r_i\}$) and, 
thus, the only payoff vector in the Core is $((p+q) \delta, 0, 
\dots, 0)$, i.e. the seller gets everything and the recommenders
get nothing.
%
% SUBSECTION: SHAPLEY VALUE
%
\subsubsection{Shapley Value}\label{sec:many-shapley}
Recall that the Shapley value assigns to each game a unique payoff vector that is ``fair'' as it satisifies the {\em Symmetry}, {\em Dummy}, and {\em Additivity} axioms. For a formal definition of the Shapley value (and related definitions) see Section~\ref{sec:one-shapley}.
\begin{theorem}\label{the:shapley-many}
Consider the game $\langle N, v \rangle$ (General). The Shapley value $\phi(v) = (\phi_s(v), \phi_{r_1}(v), \dots, \phi_{r_n}(v))$ is given by
\begin{align*}
&\phi_s(v)     = p \delta + \sum_{S \subseteq N \setminus \{s\}} \frac{|S|!(|N|-1-|S|)!}{|N|!} f(S \cup \{s\}) \delta\\
&\phi_{r_i}(v) = \sum_{\substack{S \subseteq N \setminus \{r_i\}:\\s \in S}} \frac{|S|!(|N|-1-|S|)!}{|N|!} (f(S \cup \{r_i\}) - f(S))\delta.
\end{align*}
\end{theorem}
\begin{proof}
The claim follows from the definition of the game $\langle N, v \rangle$ (General) and Definition~\ref{def:shapley}. The worths $v(S)$ of all $S \subseteq N = \{s,r_1,r_2\}$ are given explicitly in Table~\ref{tab:worths}. For $|N| > 3$  the worths $v(S)$ are obtained similarly.
\end{proof}

For the game $\langle N, v \rangle$ (Linear) this means that
\begin{align*}
&\phi_s(v)     = p \delta + \frac{1}{2}\sum_{i} q_i \delta  && \mbox{and}\\
&\phi_{r_i}(v) = \frac{1}{2} q_i \delta                     && \mbox{for all} \ i.
\end{align*}
%
% This can be seen as follows:
% \phi_s(v) &= p \delta + \sum_{S \subseteq N \setminus \{s\}} \frac{|S|!(|N|-1-|S|)!}{|N|!} f(S \cup \{s\}) \delta \\
%            &= p \delta + \sum_{S \subseteq N \setminus \{s\}} \frac{|S|!(|N|-1-|S|)!}{|N|!} \sum_{r_i \in S} q_i \delta \\
%            &= p \delta + \sum_{r_i \in N} q_i \sum_{j=1}{n} \frac{(n-1)!}{(n-j)!(j-1)!} \frac{j!(n-j)!}{(n+1)!}  \delta \\
%            &= p \delta + \sum_{r_i \in N} q_i \frac{1}{n(n+1)} \sum_{j=1}{n} j  \delta \\
%            &= p \delta + \frac{1}{2} \sum_{r_i \in N} q_i \delta.
% and
% \phi_{r_i}(v) & = \sum_{S \subseteq N \setminus \{r_i\}:s \in S} \frac{|S|!(|N|-1-|S|)!}{|N|!} (f(S \cup \{r_i\}) - f(S)) \delta \\
%               & = \sum_{S \subseteq N \setminus \{r_i\}:s \in S} \frac{|S|!(|N|-1-|S|)!}{|N|!} q_i \delta \\
%               & = \sum_{j=1}{n} \frac{(n-1)!}{(n-j)!(j-1)!} \frac{j!(n-j)!}{(n+1)!} q_i \delta \\
%               & = \frac{1}{n(n+1)} \sum_{j=1}{n} j q_i \delta \\
%               & = \frac{1}{2} q_i \delta.

For the game $\langle N, v \rangle$ (Threshold) this means that
\begin{align*}
&\phi_s (v)     = p \delta + \left(1-n\frac{k!(n-k)!}{(n+1)!}\right) q \delta && \mbox{and}\\
&\phi_{r_i} (v) = \frac{k!(n-k)!}{(n+1)!} q \delta               && \mbox{for all} \ i.
\end{align*}
%
%\begin{corollary}[Linear]
%Consider the game $\langle N, v \rangle$ (Linear). The Shapley value $\phi(v) = (\phi_s(v), \phi_{r_1}(v), \dots, \phi_{r_n}(v))$ is given by
%\begin{align*}
%&\phi_s (v)     = \delta p + \frac{1}{2}\delta\sum_{i} q_i, && \mbox{and}\\
%&\phi_{r_i} (v) = \frac{1}{2} \delta q_i                           && \mbox{for all i}.
%\end{align*}
%\end{corollary}
%
%\begin{corollary}[Threshold]
%Consider the game $\langle N, v \rangle$ (Threshold). The Shapley value $\phi(v) = (\phi_s(v), \phi_{r_1}(v), \dots, \phi_{r_n}(v))$ is given by
%\begin{align*}
%&\phi_s (v)     = \delta p + (1-n\frac{k!(n-k)!}{(n+1)!})\delta q, && \mbox{and}\\
%&\phi_{r_i} (v) = \frac{k!(n-k)!}{(n+1)!} \delta q                 && \mbox{for all i}.
%\end{align*}
%\end{corollary}
%
This suggests that in the game $\langle N, v \rangle$ (Linear) each individual recommender should receive a share of exactly one half of her contribution to the expected ``extra profit'' of the recommender. For the game $\langle N, v \rangle$ (Threshold) the fraction $k!(n-k)!/(n+1)!$ is exactly the fraction of times where this recommender's recommendation ``makes a difference''. So all in all the Shapley value does not only give a unique payoff vector, but it also yields ``fair'' payoffs in the sense that the payoff to each recommender is proportional to the recommender's contribution to the ``extra profit'' the seller can expect.
\subsubsection{Anonymity-Proof Shapley Value}\label{sec:many-proofshapley}
What would be a ``fair'' payoff vector if each recommendation was a {\em collection of arguments}? A straightforward approach would be to compute the Shapley value on the basis of arguments and to redeem recommender $r_i$ with $\sum_{a \in S_i} \phi_a(v)$, where $a$ is an argument from the set of arguments $A$ and $S_i$ is the set of arguments that recommender $r_i$ possesses; the sets $S_i$ being disjoint. The problem with this approach, however, is that it might be beneficial for a recommender to withhold some of her arguments:
%
%In practice recommenders typically have arguments in favor of a product. E.g. ``I bought this product a while ago and it really does a good job.'' What is a fair price for a recommendation in this setting?

\vspace{6pt}  {\em Example 1.}
Let $A = \{a,b,c\}$, $v(\{a,b\}) = v(\{a,c\}) = v(\{a,b,c\}) = 1$, and $v(\{a\}) = v(\{b\}) = v(\{c\}) = v(\{b,c\}) = 0.$
%\begin{enumerate}
%\item[$\bullet$]$v(\{a,b\}) = v(\{a,c\}) = v(\{a,b,c\}) = 1$,
%\item[$\bullet$]$v(\{a\}) = v(\{b\}) = v(\{c\}) = v(\{b,c\}) = 0.$
%\end{enumerate}
Let $S_1 = \{a\}$ and $S_2 = \{b,c\}.$ Then $r_1$ gets $\phi_a(v) = \frac{1}{2}$ and $r_2$ gets $\phi_b(v) + \phi_c(v) = \frac{1}{6} + \frac{1}{6} = \frac{1}{3}.$

\vspace{6pt} {\em Example 2.}
Let $A' = \{a,b\}$, $v(\{a,b\}) = 1$, and $v(\{a\}) = v(\{b\}) = 0.$
%\begin{enumerate}
%\item[$\bullet$]$v(\{a,b\}) = 1$,
%\item[$\bullet$]$v(\{a\}) = v(\{b\}) = 0.$
%\end{enumerate}
Let $S_1 = \{a\}$ and $S_2 = \{b\}.$ Then $r_1$ gets $\phi_a(v) = \frac{1}{2}$ and $r_2$ gets $\phi_b(v) = \frac{1}{2}.$ I.e.~$r_2$ would be better off.

\vspace{6pt} \noindent The {\em anonymity-proof Shapley value}~\cite{OCS08} cannot be ``tricked'' in this way. It is defined as follows:
%
% PD: REMOVED THIS
% One can show that there is no anonymity-proof value that satisfies the {\em Symmetry}, {\em Dummy}, and {\em Additivity} axiom (see Section~\ref{sec:one-shapley}. But one can also show that there is a unique {\em best approximate monotone} value that satisfies relaxed versions of these axioms. Any value that is best approximate monotone is also anonymity proof, but not vice versa. Thus, best approximate monotonicity is stronger than anonymity proofness.\footnote{The uniqueness is lost if one requires anonymity proofness instead of best approximate monotonicity.} 
%
\begin{definition}
For any set $A' \subseteq A$ of declared arguments the {\em anonymity-proof Shapley value} $\psi_a(v)$ for $a \in A'$ is:
\begin{align*}
&\psi_a(v) = \frac{\phi_a(v)}{\sum_{a' \in A'} \phi_{a'}(v)} v(A').
\end{align*}
\end{definition}   
So a better way to redeem the recommenders would be to compute the anonymity-proof Shapley value $\psi_a(v)$ for each argument $a \in A'$ and to give each recommender $\sum_{a \in S_i} \psi_a(v).$ With this approach $r_1$ and $r_2$ would get $\psi_a = 3/5$ and $\psi_b(v) + \psi_c(v) = 2/5$ in Example 1 and $\psi_a = 3/4$ and $\psi_b(v) = 1/4 < 2/5$ in Example 2.
%
%\vspace{6pt}{\em Example 1 (cont'd).}
%Since $A' = \{a,b,c\}$, we have
%\begin{enumerate}
%\item[$\bullet$] $\psi_a = \frac{\phi_a(v)}{\sum_{a' \in A'} \phi_{a'}(v)} \cdot v(A')= \frac{1/2}{1/2+1/6+1/6} \cdot 1 = \frac{3}{5}$
%\item[$\bullet$] $\psi_b = \frac{\phi_b(v)}{\sum_{a' \in A'} \phi_{a'}(v)} \cdot v(A')= \frac{1/6}{1/2+1/6+1/6} \cdot 1 = \frac{1}{5}$
%\end{enumerate}
%And so $r_1$ gets $\psi_a = 3/5$ and $r_2$ gets $\psi_b(v) + \psi_c(v) = 2/5.$
%
%\vspace{6pt}{\em Example 2 (cont'd).}
%Since $A' = \{a,b\}$, we have
%\begin{enumerate}
%\item[$\bullet$] $\psi_a = \frac{\phi_a(v)}{\sum_{a' \in A'} \phi_{a'}(v)} \cdot v(A')= \frac{1/2}{1/2+1/6} \cdot 1 = \frac{3}{4}$
%\item[$\bullet$] $\psi_b = \frac{\phi_b(v)}{\sum_{a' \in A'} \phi_{a'}(v)} \cdot v(A')= \frac{1/6}{1/2+1/6} \cdot 1 = \frac{1}{4}$
%\end{enumerate}
%And so $r_1$ gets $\psi_a = 3/4$ and $r_2$ gets $\psi_b(v) = 1/4 < 2/5.$
%
\subsubsection{The Nash Bargaining Solution}\label{sec:many-nash}
Recall that the Nash Bargaining Solution is defined as the unique bargaining solution that satisfies the axioms listed in Section~\ref{sec:one-nash}. 
%For a formal definition of the Nash Bargaining Solution (and related definitions) see Section~\ref{sec:one-nash}.
%
\begin{theorem}
For $\langle N, v \rangle$ (General) let $F = \{(f_s,f_{r_1},$ $\dots$, $f_{r_n})$ $| f_s \ge 0$, $f_{r_i} \ge 0 \ (\forall \ i) \ \text{and} \ f_s + \sum_{i} f_{r_i} = (p + f(N)) \cdot \delta\}$ and $d = (d_s, d_{r_1}, \dots, d_{r_n}) = (p \cdot \delta, 0, \dots, 0).$ Then,
\begin{align*}
\phi_s(F,d) & = ( p + \frac{1}{n+1} f(N) ) \cdot \delta &&\text{and}\\
\phi_{r_i}(F,d) &= \frac{1}{n+1} f(N) \cdot \delta. &&\text{for all} \ i.
\end{align*}
\end{theorem}
\begin{proof}
The claim follows from the fact that $\phi(F,d) \in \mbox{argmax}_{f \in F} \prod_i (f_i - d_i)$~\cite{M97}.
\end{proof}
One problem with the Nash Bargaining Solution is that it completely ignores the possibility of cooperation among subsets of players. To see that this may lead to ``unfair'' prices, consider the game $\langle N, v \rangle$ (Linear) with one seller $s$ and two recommenders $r_1$ and $r_2.$ Suppose that $q_1 = 1 - p - \epsilon$ and that $q_2 = \epsilon$ for some some small $\epsilon > 0$. It follows that $\phi_{r_1}(F,d) = \phi_{r_2}(F,d) = \frac{1-p}{2}$, i.e. the expected payoff to both recommenders is the same. But  since $r_1$'s contribution to the expected worth of the grand coalition is significantly higher than that of $r_2$ (especially as $\epsilon \to 0$), this cannot be regarded as ``fair''. We conclude that for $|N| > 2$ it is not advisable to use the Nash Bargaining Solution to guide the pricing of recommendations.
%
%% SECTION: Profit Maximization %%
%\section{Part II: Recommending Strategically}\label{sec:recommending-strategically}
\section{Recommending Strategically}\label{sec:recommending-strategically}
We study the following problem: There are $n$ products. For each product the recommender has two options: ``recommend it'' or ``\emph{not} recommend it''. A recommendation is {\em successful} if the buyer buys the product. For a successful recommendation the recommender gets a constant reward of $r$ and this reward is the same for all products. Initially, the probability $p$ of success is $p_0 < 1$. With each unsuccessful recommendation this probability drops from its current value to $p=l\cdot p$, where $l<1$ is the {\em loss rate}. The probability $p$ can be seen as an estimate of the recommendee's \emph{trust} in the recommender and a high value of $l$ corresponds to a slow loss in trust. This basic model is analyzed in Section \ref{sec:wo_reset_wo_recovery}. We also consider extensions of this model where trust (= $p$) can increase again in two ways. First, we assume that $p$ is reset to $p=p_0$ on each successful recommendation. This setting we refer to as ``with reset'' and it is analyzed in Section \ref{sec:with_reset_wo_recovery}. Second, we introduce a factor $g\ge 1$ and each time the recommender does not recommend anything trust is regained and $p$ is updated to $p=\min(g\cdot p,p_0)$. This setting we refer to as ``with recovery'' when $g>1$ and it is analyzed in Section \ref{sec:with_reset_with_recovery}.

In all settings the recommender's sole goal is to maximize the overall expected reward $M_n(p_0,l,g)$ for the given parameters $p_0$, $l$ and $g.$ We are interested in the asymptotic behavior of $M_n(p_0,l,g)$, i.e. in $R(p_0,l,g) = \lim_{n \to \infty} M_n(p_0,l,g).$ 
% For $n = 1:$ $M_1(p_0,l) = p_0 \cdot r.$ 
% For $n > 1:$ $M_n(p,l) = p (r + M(p_0,l) )+ (1-p)(M(lp,l))$, where $p$ is the current probability of success.
Before looking at the theoretical analysis, the following section experimentally demonstrates the different behavior of the optimal total expected reward in these settings.
\subsection{Experimental Results}\label{sec:experiments}
Figure~\ref{fig:experiments} gives experimental results for $n = 200$, $r = 1$, $p_0 = 0.5$, $l = 0.66$, $g = 1$ (in the setting ``without recovery'') and $g=1.33$ (in the setting ``with recovery''). It shows that the expected reward of the optimal strategy converges in the setting ``without recovery'' and diverges in the setting ``with recovery''. In the setting ``without recovery'' the expected reward converges to $2.25$ if the probability of success is not reset and to $5$ if it is reset to $p_0$ on a single successful recommendation. The figure also shows that the expected reward of the heuristic ``recommend product $1$, $k+1$, $2k+1$, etc.'' converges for $k = 2$ where $l \cdot g^k < 1$ and diverges for $k = 3$ and $4$ where $l \cdot g^k > 1$. Finally, it shows that the expected reward grows faster for $k = 3$ than for $k = 4.$  
\begin{figure}[ht!] %changed h! to ht! to turn off warning
\centering
\epsfig{file=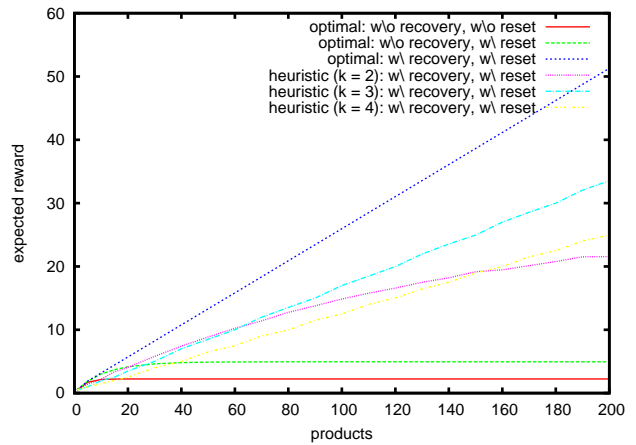,width=6cm,angle=270}
%\vspace{-3mm}
\caption{Without any recovery (red and green lines) the total expected reward converges. This also holds, even with recovery, for the ``aggressive'' heuristic (purple line) which recommends every second item. The other three settings, described in detail in Section~\ref{sec:experiments}, lead to an asymptotically unbounded reward.}\label{fig:experiments}
\end{figure}
\subsection{Without Reset, without Recovery}\label{sec:wo_reset_wo_recovery}
Here we consider the case where $g=1$ (= no recovery) and assume that the probability of success is {\em not} reset to $p_0$ on a successful recommendation. As the probability of success remains unchanged if no recommendation is given, the optimal strategy is to recommend {\em all} products. Therefore we can rewrite $R(p_0,l):=R(p_0,l,1)$ as follows: $R(p_0,l) = p_0 (r + R(p_0,l)) +  (1-p_0) R(lp_0,l)$, which we can solve analytically.
\begin{align*}
R(p_0,l) &=   \frac{p_0}{1-p_0} r + R(p_0 \cdot l,l)\\
         &= \frac{p_0}{1-p_0} r + \frac{p_0 \cdot l}{1-p_0} r + R(p_0 \cdot l^2,l)\\         
         &= \frac{p_0}{1-p_0} \cdot \sum_{i=0}^{\infty} l^i \cdot r =   \frac{p_0}{1-p_0} \cdot \frac{1}{1 - l} \cdot r < \infty.
%R(p_0,l) &=   \frac{p_0}{1-p_0} r + R(p_0 \cdot l,l) \\
%         &= \frac{p_0}{1-p_0} r + \frac{p_0 \cdot l}{1-p_0} r + R(p_0 \cdot l^2,l)\\         
%         &= \frac{p_0}{1-p_0} \cdot \sum_{i=0}^{\infty} l^i \cdot r \\
%         &=   \frac{p_0}{1-p_0} \cdot \frac{1}{1 - l} \cdot r < \infty.
\end{align*}

So, not surprisingly, if trust can only be lost and if both the initial trust $p_0$ and the loss rate $l$ are strictly smaller than $1$, then the total expected reward the recommender can achieve is finite, even when there is an infinite sequence of items to recommend.

\subsection{With Reset, without Recovery}\label{sec:with_reset_wo_recovery}
Now let us analyze the case where still $g = 1$ (= no recovery) but each successful recommendation leads to reset of $p$ to $p_0$. Again, the optimal strategy is to recommend \emph{all} products as there is no gain from not recommending. In this setting, we can rewrite $R(p_0,l)$ as follows: $R(p_0,l) = (1-q) \cdot (r + R(p_0,l))$, where $q$ denotes the probability that there will be not a single successful recommendation over the infinite sequence. This recurrence can be solved (i.e. $\lim_{n\rightarrow \infty} M_n(p_0,l)$ is finite) if and only if $q > 0.$ 
\begin{lemma}\label{lem:worecov}
Let $\mbox{dilog(x)} = \int_{x}^{1} \frac{\ln(t)}{1-t} \, dt$ and $c = \max(p_0,l)$. Then, for all $1 > p_0 \ge 0$,
$q  \geq (1-c)\mbox{exp}\left(\frac{\mbox{dilog}(1-c)}{\ln(c)}\right) > 0.$
%\begin{align*}
%q & \geq \mbox{exp}\left(\frac{\mbox{dilog}(1-c)}{\ln(c)}\right) > 0.
%\end{align*}
\end{lemma}
\begin{proof}
The probability that there will be not a single successful recommendation is:
\begin{align*}
q &= \prod_{k=0}^{\infty} (1-l^k \cdot p_0)\ge \prod_{k=0}^{\infty} (1-c^{k+1}).
\end{align*}
Hence it suffices to show that $\prod_{k=0}^{\infty} (1-c^{k+1}) > 0$. Taking the $\ln(\ )$ of both sides we get
\begin{align*}
\ln\left( \prod_{k=0}^{\infty} (1-c^{k+1}) \right) &= \sum_{k=0}^{\infty} \ln (1-c^{k+1}) > - \infty,
\end{align*}
where we need to prove this inequality. Note that the expression $\ln (1-c^{k+1})$ is strictly increasing in $k$ and
hence $\ln (1-c^{k+1}) \ge \int_{k-1}^k \! \ln (1-c^{x+1}) \, dx$. This gives the bound
%
%\begin{tabular}{c}
\begin{align*}
\sum_{k=0}^{\infty} \ln (1-c^{k+1})&= \ln (1-c) + \sum_{k=1}^{\infty} \ln (1-c^{k+1}) \\
&\ge \ln(1-c) + \sum_{k=1}^{\infty} \int_{x=k-1}^{k} \! \ln (1-c^{x+1}) \, dx \\
&= \ln(1-c) + \sum_{k=0}^{\infty} \int_{x=k}^{k+1} \! \ln (1-c^{x+1}) \, dx \\
&= \ln(1-c) + \int_{x=0}^{\infty} \! \ln (1-c^{x+1}) \, dx.
\end{align*}
%\end{tabular}
%
Recall that $\mbox{dilog}(x) = \int_{x}^{1} \frac{\ln(t)}{1-t} \, dt.$ The indefinite integral 
of $\ln(1-x)$ is $-\mbox{dilog}(1-x)/\ln(x).$ We get
%For $0<t<1$ we have $\ln(t) < 0$ and $1-t > 0.$ Hence $\mbox{dilog}(x) > 0$ for $0 < x < 1.$ Then, 
\begin{align*}
\int_{x=0}^{\infty} \! \ln (1-c^{k+1}) \, dx &= \frac{\mbox{dilog}(1 - c)}{\ln(c)} - \displaystyle{\lim_{x\to\infty}} \frac{\mbox{dilog}(1 - c^{x+1})}{\ln(c)}.
\end{align*}
Since $\mbox{dilog}(x)$ is continuous\footnote{It is even differentiable as it is defined as an indefinite integral.} and $\mbox{dilog}(1) = 0$ (see Lemma \ref{lem:dilog_bounds}), we get
\begin{align*}
\int_{x=0}^{\infty} \! \ln (1-c^{k+1}) \, dx &= \frac{\mbox{dilog}(1 - c)}{\ln(c)} - \frac{\mbox{dilog}(1)}{\ln(c)}\\
                                             &= \frac{\mbox{dilog}(1 - c)}{\ln(c)}.
\end{align*}
%\begin{align*}
%\int_{x=0}^{\infty} \! \ln (1-c^{k+1}) \, dx = \frac{\mbox{dilog}(1 - c)}{\ln(c)} - \frac{\mbox{dilog}(1)}{\ln(c)} = \frac{\mbox{dilog}(1 - c)}{\ln(c)}.
%\end{align*}
%
% For $0 \le x \le 1$ we have ...
For $0 < x < 1$ we have $0 \le \mbox{dilog}(1-x) < 2e^{-1} + 1$ (see Lemma~\ref{lem:dilog_bounds}). For $0 < x < 1$ we have $\ln(x) < 0.$ It follows that $\int_{x=0}^{\infty} \! \ln (1-c^{k+1}) \, dx > - \infty$. \qed
%\begin{align*}
%\int_{x=0}^{\infty} \! \ln (1-c^{k+1}) \, dx &> - \infty. \qed
%\end{align*}
\end{proof}
\begin{lemma}\label{lem:dilog_bounds}
Let $\mbox{dilog}(x) = \int_{1}^{x} \frac{\ln(t)}{1-t} \, dt.$ Then $\mbox{dilog}(x)$ is \hyphenation{mono-tonous-ly} monotonously decreasing and
$0 \le \mbox{dilog}(x) < 2e^{-1} + 1.$
%\begin{align*}
%&0 \le \mbox{dilog}(x) < 2e^{-1} + 1.
%\end{align*}
\end{lemma}
In fact, the tight upper bound of $\mbox{dilog}(x)\le \pi^2/6 < 2e^{-1} + 1$ is known \cite{AAR01}, but we choose to give the following elementary proof of Lemma \ref{lem:dilog_bounds} to have a self-contained argument.
\begin{proof}
Let $f(t) = -\ln(t)/(1-t)$. Then $f'(t) = -(\frac{1}{t}(1-t)+\ln(t))/(1-t)^2 < 0$ for $0<t<1$.
So $\int_{t=x}^1 f(t) dt < \int_{t=x}^{e^{-1}} f(t) dt + (1-e^{-1}) f(e^{-1})$. For $0<t\le e^{-1}$ we also have
$f(t) \le -\ln(t)/(1-e^{-1})$. So, $\int_{t=x}^{e^{-1}} f(t) dt \le [t - t\cdot \ln(t)]^{e^{-1}}_{x}$. This is largest when $x\rightarrow 0$ where the whole expression becomes $2e^{-1}$ and so $\int_{t=x}^1 f(t) dt < 2e^{-1} + 1$ for $0\le x < 1$. Note that $f(t)$ is continuous at $t=1$ with $\lim_{x\rightarrow 1} f(t) = 1$ (using e.g. the l'Hopital Rule). So trivially $\mbox{dilog}(1) = 0$. As $f(t) > 0$ this gives the desired lower bound. %of $0 \le \mbox{dilog}(x)$.
\end{proof}
Using Lemma~\ref{lem:worecov} we can prove the following theorem.
\begin{theorem}\label{the:wo-recov}
Let $\mbox{dilog(x)} = \int_{x}^{1} \frac{\ln(t)}{1-t} \, dt$, $c = \max(p_0,l)$, and $\delta(c) = (1-c) \mbox{exp}(\mbox{dilog}(1-c)/\ln(c))$. Then, for all $1 > p_0 \ge 0$,
\begin{align*}
R(p_0,l) & \le \frac{1-\delta(c)}{\delta(c)} \cdot r < \infty.
\end{align*}
\end{theorem}
This proves that even if the probability of success is reset to $p_0$ on a single successful recommendation, the total expected reward over an infinite period is bounded.
\subsection{With Reset, with Recovery}\label{sec:with_reset_with_recovery}
Finally, we consider the setting where $g > 1$ (= with recovery). Here the probability of success is set to $\min(p_0,g \cdot p)$ if no recommendation was given. Hence it might be better not to recommend all products to avoid that that probability $p$ converges to zero. Let $M_n(p_0,l,g)$ denote the expected reward of the optimal strategy. To obtain bounds for $M_n(p_0,l,g)$, let us consider, as a heuristic, the algorithm  $A^{(k)}$ that recommends product $1$, $k+1$, $2k+1$, etc. We write $A^{(k)}_n(p_0,l,g)$ to denote this algorithm's expected profit.
\begin{theorem}\label{the:w-recov-1}
Let $\psi$ be the smallest integer such that $l \cdot g^{\psi} \ge 1.$ If $k > \psi$, then, for all $1 > p_0, l > 0$ and $\infty > g \ge 1$, 
\begin{align*}
A^{(k)}_n(p_0,l,g) &= \lfloor \frac{n}{k} \rfloor \cdot p_0 \cdot r.
\end{align*}
\end{theorem}
\begin{proof}
The expected reward for the first recommendation is $p_0 \cdot r.$ Since $k > \psi$, the expected reward for every other recommendation is also $\min(p_0, p_0 \cdot l \cdot g^{k-1}) = p_0 \cdot r.$ Since there are exactly $\lfloor n/k \rfloor$ recommendations, this shows that $A^{(k)}_n(p_0,l,g) = \lfloor \frac{n}{k} \rfloor \cdot p_0 \cdot r.$ 
\end{proof}
This is instructive as it shows that (a) for $k > \psi$ the expected reward $A^{(k)}_n(p_0,l,g)$ of $A^{(k)}$ does {\em not} converge as $n$ tends to infinity and (b) for $k' > k > \psi$ the expected reward $A^{(k')}_n(p_0,l,g)$ of $A^{(k')}$ grows slower (and is ultimately lower) than the expected reward $A^{(k)}_n(p_0,l,g)$ of $A^{(k)}.$ 
%Of course, this also shows non-convergence if the algorithm recommends items $c$, $c + k$, $c + 2k$ ... for $c > 1$ and $k \ge \psi$. The expected reward is then $\lfloor \frac{n-c+1}{k'+1} \rfloor \cdot r \cdot p_0$.
Since the reward $M_n(p_0,l,g)$ of the optimal strategy is at least as high, this also shows non-convergence of $R(p_0,l,g) = \lim_{n \to \infty} M_n(p_0,l,g) = \infty.$
\begin{theorem}\label{the:w-recov-2}
Let $\psi$ be the smallest integer such that $l \cdot g^\psi \ge 1.$ If $k \le \psi$, then, for all $1 > p_0, l > 0$ and $\infty > g \ge 1$, there exist $p'_0$ and $l'$ such that
\begin{align*}
\lim_{n \rightarrow \infty} A^{(k)}_n(p_0,l,g) &\le \lim_{n \rightarrow \infty} M_n(p'_0,l') < \infty.
\end{align*}
\end{theorem}
\begin{proof}
If $k < \psi$, then the profit maximization problem with parameters $p_0$, $l$, and $g$ on the products $1$, $2$, $3$, etc. is equivalent to the profit maximization problem with parameters $p'_0 = p_0$, $l' = l \cdot g^{k-1} < 1$, and $g' = 1$ on the products $1$, $k+1$, $2k+1$, etc. The claim follows from Theorem~\ref{the:wo-recov}.
\end{proof}
Whereas Theorem~\ref{the:w-recov-1} shows that recommending too seldomly is sub-optimal, Theorem~\ref{the:w-recov-2} shows that recommending too often is even worse.% When pushed to the limit ($k\rightarrow 1$), this removes any chance for recovery and leads to a convergence of $p$ to zero.
\section{Discussion and Future Work}

Suppose you recommend a product to a friend and the seller of the product offers 
to pay you for your recommendation. What would be a ``good'' price? Our first
finding was that the only ``truthful'' price would be zero. The problem with this,
however, is that if you do not get paid, then you might as well decide {\em not} to 
recommend the product. And so the seller might be willing to pay you a ``fair'' 
price. We approached the problem of finding ``fair'' prices by studying solution
concepts from coalitional game theory such as the Core, the Shapley value, and the
Nash Bargaining Solution. Since each of these solution concepts formalizes some notion 
of ``fairness'', these prices can be regarded as {\em provably} ``fair''. We view such 
an ``axiomatic'' foundation of ``fairness'' to be the only viable basis for truely 
``fair'' prices in practice.

Now suppose that you get paid for each succesful recommendation you make, and that 
you want to maximize the amount of money paid to you. At first sight, it might appear 
that the best strategy for you is to send out as many recommendations to as many 
friends as possible. But, then, just as you get ``blind'' when being shown too many 
ads, your friends will probably start to ignore your ``recommendations''. We adressed 
this problem by modeling the loss in ``trust'' by a drop in ``purchase probability'' 
on each unsuccesful recommendation. Our main finding here was that, even if the ``trust'' 
in you is reset to the initial level on a single successful recommendation, the total 
expected profit you can make over an infinite period of time is bounded. This can only 
be overcome if the recomendee also incrementally regains ``trust'' over periods without 
any recommendation. 

% We see this as closely related to ``banner blindness''~\cite{BL98,CHN00,BHNG05}, where 
% people have become so overloaded and fed up with advertisement that they stop to notice 
% it completely. Seen from this angle, our findings indicate that advertisers might have to 
% {\em stop} showing advertisements on a regular basis if they want to retain customers' 
% ``trust'' without seeing CTR's converge to zero.

We believe that our work motivates a number of interesting research questions. E.g.: What 
are ``good'' pricing mechanisms in settings where the seller has objectives such as maintaining 
the buyer's ``trust''? How exactly do web users respond to being shown irrelevant advertisements? 
Is it possible to revive their interest in banner ads? What are ``optimal'' auction mechanisms 
for sponsored search when the click-through-rates are non-constant and decay with each irrelevant 
advertisement being shown?

%\vspace{-1.3mm}

%% BIBLIOGRAPHY %%
\bibliographystyle{abbrv}
\bibliography{pricing_recommendations}

\begin{thebibliography}{10}

\bibitem{AMT07}
Z.~Abrams, O.~Mendelevitch, and J.~Tomlin.
\newblock Optimal delivery of sponsored search advertisements subject to budget
  constraints.
\newblock In {\em Conference on Electronic commerce (EC'07)}, pages 272--278,
  2007.

\bibitem{AAR01}
G.~E. Andrews, R.~Askey, and R.~Roy.
\newblock {\em Special functions}.
\newblock Cambridge University Press, 2001.

\bibitem{AMSX09}
D.~Arthur, M.~Motwani, A.~Sharma, and Y.~Xu.
\newblock Pricing strategies for viral marketing on social networks.
\newblock In {\em Workshop on Internet and Network Economics (WINE'09)}, page
  to appear, 2009.

\bibitem{BL98}
J.~P. Benway and D.~M. Lane.
\newblock Banner blindness: Web searchers often miss ``obvious'' links.
\newblock {\em ITG Newsletter}, 1(3), 1998.
\newblock
  \url{http://www.internettg.org/newsletter/dec98/banner_blindness.html}.

\bibitem{BO06}
D.~Bergemann and D.~Ozmen.
\newblock Optimal pricing with recommender systems.
\newblock In {\em Conference on Electronic commerce (EC'06)}, pages 43--51,
  2006.

\bibitem{B63}
O.~N. Bondareva.
\newblock Some applications of linear programming methods to the theory of
  cooperative games.
\newblock {\em Problemy Kybernetiki}, 10:119--139, 1963.

\bibitem{BR87}
J.~J. Brown and P.~H. Reingen.
\newblock Social ties and word-of-mouth referral behavior.
\newblock {\em Journal of Consumer Research: An Interdisciplinary Quarterly},
  14(3):350--62, 1987.

\bibitem{BHNG05}
M.~Burke, A.~Hornof, E.~Nilsen, and N.~Gorman.
\newblock High-cost banner blindness: Ads increase perceived workload, hinder
  visual search, and are forgotten.
\newblock {\em ACM Transactions on Computer-Human Interactaction},
  12(4):423--445, 2005.

\bibitem{CHN00}
P.~Chatterjee, D.~L. Hoffman, and T.~P. Novak.
\newblock Modeling the clickstream: Implications for web-based advertising
  efforts.
\newblock {\em Marketing Science}, 22:520--541, 2000.

\bibitem{DR01}
P.~Domingos and M.~Richardson.
\newblock Mining the network value of customers.
\newblock In {\em SIGKDD international conference on Knowledge discovery and
  data mining (KDD'01)}, pages 57--66, 2001.

\bibitem{FMPS07}
J.~Feldman, S.~Muthukrishnan, M.~Pal, and C.~Stein.
\newblock Budget optimization in search-based advertising auctions.
\newblock In {\em Conference on Electronic commerce (EC'07)}, pages 40--49,
  2007.

\bibitem{FN09}
J.~Fernandez and B.~Nahata.
\newblock Pay what you like.
\newblock Technical Report 16265, Munich Personal RePEc Archive, 2009.

\bibitem{FV08}
Friend vouch, 2008.
\newblock \url{http://www.friendvouch.com}.

\bibitem{GEB04}
A.~C.~B. Garcia, M.~Ekstrom, and H.~Bj\"{o}rnsson.
\newblock Hyriwyg: leveraging personalization to elicit honest recommendations.
\newblock In {\em Conference on Electronic commerce (EC'04)}, pages 232--233,
  2004.

\bibitem{G59}
D.~Gillies.
\newblock {\em Contributions to the Theory of Games IV}, chapter Solutions to
  general non-zero-sum games, pages 47--–85.
\newblock Princeton University Press, 1959.

\bibitem{GCD03}
R.~Grewal, T.~W. Cline, and A.~Davies.
\newblock Early-entrant advantage, word-of-mouth communication, brand
  similarity, and the consumer decision-making process.
\newblock {\em Journal of Consumer Psychology}, 13(3):187--197, 2003.

\bibitem{HKTR04}
J.~L. Herlocker, J.~A. Konstan, L.~G. Terveen, and J.~T. Riedl.
\newblock Evaluating collaborative filtering recommender systems.
\newblock {\em ACM Transactions on Information Systems}, 22(1):5--53, 2004.

\bibitem{HKK91}
P.~M. Herr, F.~R. Kardes, and J.~Kim.
\newblock Effects of word-of-mouth and product-attribute information on
  persuasion: An accessibility-diagnosticity perspective.
\newblock {\em Journal of Consumer Research}, 17(4):454--462, 1991.

\bibitem{KKT03}
D.~Kempe, J.~Kleinberg, and E.~Tardos.
\newblock Maximizing the spread of influence through a social network.
\newblock {\em KDD}, pages 137--146, 2003.

\bibitem{L06}
S.~Lahaie.
\newblock An analysis of alternative slot auction designs for sponsored search.
\newblock In {\em Conference on Electronic commerce (EC'06)}, pages 218--227,
  2006.

\bibitem{LPSV07}
S.~Lahaie, D.~Pennock, A.~Saberi, and R.~Vohra.
\newblock {\em Algorithmic Game Theory}, chapter Sponsored Search Auctions,
  pages 699--716.
\newblock Cambridge University Press, 2007.

\bibitem{LAH06}
J.~Leskovec, L.~A. Adamic, and B.~A. Huberman.
\newblock The dynamics of viral marketing.
\newblock In {\em Conference on Electronic commerce (EC'06)}, pages 228--237,
  2006.

\bibitem{M08b}
A.~Mantzaris.
\newblock Pay-what-you-like restaurants, 2008.
\newblock
  \url{http://www.budgettravel.com/bt-dyn/content/article/2008/02/29/AR2008022%
902761.html}.

\bibitem{M08a}
S.~Maxwell.
\newblock {\em The Price is Wrong: Understanding What Makes a Price Seem Fair
  and the True Cost of Unfair Pricing}.
\newblock John Wiley and Sons, 2008.

\bibitem{M91}
H.~Moulin.
\newblock {\em Axioms of Cooperative Decision Making (Econometric Society
  Monographs)}.
\newblock Cambridge University Press, July 1991.

\bibitem{M97}
R.~B. Myerson.
\newblock {\em Game Theory: Analysis of Conflict}.
\newblock Harvard University Press, 1997.

\bibitem{N50}
J.~Nash, J.~F.
\newblock The bargaining problem.
\newblock {\em Econometrica}, 18(2):155--162, 1950.

\bibitem{OCS08}
N.~Ohta, V.~Conitzer, Y.~Satoh, A.~Iwasaki, and M.~Yokoo.
\newblock Anonymity-proof shapley value: Extending shapley value for
  coalitional games in open environments.
\newblock {\em Autonomous Agents and Multiagent Systems}, pages 927--934, 2008.

\bibitem{OR94}
M.~J. Osborne and A.~Rubinstein.
\newblock {\em A Course in Game Theory}.
\newblock The MIT Press, 1994.

\bibitem{PK08}
K.~Patel and R.~G.~P. Kantamneni.
\newblock Monetizing low value clickers.
\newblock United States Patent Application 20080249854, 2008.

\bibitem{RFBS84}
P.~H. Reingen, B.~L. Foster, J.~J. Brown, and S.~B. Seidman.
\newblock Brand congruence in interpersonal relations: A social network
  analysis.
\newblock {\em Journal of Consumer Research}, 11(3):771--783, 1984.

\bibitem{R04}
J.~J. Rotemberg.
\newblock Fair pricing.
\newblock Technical Report 10915, National Bureau of Economic Research, 2004.

\bibitem{S53}
L.~S. Shapley.
\newblock {\em Contributions to the Theory of Games II}, chapter A Value for
  n-person Games, pages 307--–317.
\newblock Princeton University Press, 1953.

\bibitem{S67}
L.~S. Shapley.
\newblock On balanced sets and cores.
\newblock {\em Naval Research Logistics Quarterly}, 14:453--460, 1967.

\bibitem{SM95}
U.~Shardanand and P.~Maes.
\newblock Social information filtering: algorithms for automating ``word of
  mouth''.
\newblock In {\em SIGCHI conference on Human factors in computing systems
  (CHI'95)}, pages 210--217, 1995.

\bibitem{SW09}
A.~Singla and I.~Weber.
\newblock Camera brand congruence in the flickr social graph.
\newblock In {\em Conference on Web Search and Data Mining (WSDM'09)}, pages
  252--261, 2009.

\end{thebibliography}

\end{document}